\documentclass{article}

\newif\ifconf
\conffalse

\usepackage[utf8]{inputenc}

\usepackage[draft]{ktmacros}
\usepackage{vfmacros}
\usepackage[capitalise,nameinlink]{cleveref}

\usepackage{graphicx,fullpage,amsfonts}
\usepackage{subcaption}

\def\R{\mathcal{R}}
\def\A{\mathcal{A}}

\def\D{\mathcal{D}}
\def\reference{\rho}
\def\ID{\mathrm{ID}}
\newcommand{\bool}{\mathtt{bool}}

\newcommand{\ind}[1]{\mathtt{ind}\left(#1\right)}

\def\TV{{\rm TV}}
\def\Bern{{\rm Bernoulli}}
\def\dfn{:=}

\newcommand{\nut}{\tilde{\nu}}

\def\PrivUnit{\mathtt{PrivUnit}}
\def\decode{\mathtt{decode}}
\def\PrivHemi{\mathtt{PrivHS}}
\def\PrivHS{\mathtt{PrivHS}}
\def\PrivUnitOpt{\mathtt{PrivUnit Optimized}}

\usepackage[style=alphabetic,backend=bibtex,maxnames=12,maxbibnames=10,maxcitenames=10,maxalphanames=10,giveninits=true,doi=false,url=true]{biblatex}
\newcommand*{\citet}[1]{\AtNextCite{\AtEachCitekey{\defcounter{maxnames}{2}}} \textcite{#1}}

\newcommand*{\citep}[1]{\cite{#1}}

\bibliography{compression-refs,vf-allrefs-local}

\title{Lossless Compression of Efficient Private Local Randomizers}

\author{
  Vitaly Feldman\\
  Apple\\
  \and Kunal Talwar\\
  Apple\\
}
\date{}
\begin{document}

\maketitle

\begin{abstract}
Locally Differentially Private (LDP) Reports are commonly used for collection of statistics and machine learning in the federated setting. In many cases the best known LDP algorithms require sending prohibitively large messages from the client device to the server (such as when constructing histograms over large domain or learning a high-dimensional model). This has led to significant efforts on reducing the communication cost of LDP algorithms.

At the same time LDP reports are known to have relatively little information about the user's data due to randomization. Several schemes are known that exploit this fact to design low-communication versions of LDP algorithm but all of them do so at the expense of a significant loss in utility. Here we demonstrate a general approach that, under standard cryptographic assumptions, compresses every efficient LDP algorithm with negligible loss in privacy and utility guarantees. The practical implication of our result is that in typical applications the message can be compressed to the size of the server's pseudo-random generator seed. More generally, we relate the properties of an LDP randomizer to the power of a pseudo-random generator that suffices for compressing the LDP randomizer. From this general approach we derive  low-communication algorithms for the problems of frequency estimation and high-dimensional mean estimation. Our algorithms are simpler and more accurate than existing low-communication LDP algorithms for these well-studied problems.



\end{abstract}

\section{Introduction}
We consider the problem of collecting statistics and machine learning in the setting where data is held on a large number of user devices. The data held on devices in this \emph{federated} setting is often sensitive and thus needs to be analyzed with privacy preserving techniques. One of the key approaches to private federated data analysis relies on the use of locally differentially private (LDP) algorithms to ensure that the report sent by a user's device reveals little information about that user's data. Specifically, a randomized algorithm $\R\colon X \to Y$ is an $\eps$-DP local randomizer if for every possible output $y \in Y$, and any two possible values of user data $x_1,x_2 \in X$, $\pr[\R(x_1) =y]$ and $\pr[\R(x_2) =y]$ are within a factor of $e^\eps$  (where the probability is taken solely with respect to the randomness of the algorithm $\R$).

The concept of a local randomizer dates back to the work of \citet{Warner65} where it was used to encourage truthfulness in surveys. In the context of modern data analysis it was introduced by \citet{EvfimievskiGS03} and then related to differential privacy in the seminal work of \citet{DworkMNS:06}. Local randomizers are also used for collection of statistics and machine learning in several industrial applications \cite{erlingsson2014rappor,appledp,DKY17-Microsoft}. Practical applications such as building a histogram over a large domain or training a model with millions of parameters \citep{McMahanRT018}, require applying the randomizer to high dimensional data. Many of the standard and most accurate ways to randomize such data result in reports whose size scales linearly with the dimension of the problem. Communication from the user devices is often significantly constrained in practical applications. This limits the scope of problems in which we can achieve the best known utility-privacy trade-off and motivates significant research interest in designing communication-efficient LDP algorithms.

\subsection{Our contribution}
In this work, we explore practical and theoretical aspects of compressing outputs of LDP mechanisms. We focus on the $\eps > 1$ regime, that has been motivated by recent privacy amplification techniques based on anonymization and shuffling of LDP reports~\citep{BittauEMMR17,ErlingssonFMRTT19,CheuSUZZ19,BalleBGN19,FeldmanMT20}. It has long been noted that, by design, the output of an LDP randomizer contains a limited amount of information about the input data. Thus it should be compressible using standard tools from information theory. However, standard compression algorithms do not necessarily preserve privacy. \citet{bassily2015local} and \citet{bun2019heavy} describe general privacy preserving techniques for compressing LDP protocols. Unfortunately, their techniques result either in the loss of accuracy (as only a fraction of the user population ends up contributing a report) or an increase in $\eps$ by a constant factor which makes the approach impractical in the $\eps > 1$ regime. We remark that it is crucial to use a compression technique that preserves the privacy guarantee since in some problems accuracy scales as $e^{-\eps/2}$ when $\eps >1$. In addition, the central privacy guarantees resulting from amplification by shuffling also scale as $e^{-\eps/2}$.


We propose a general approach to compressing an arbitrary local randomizer that preserves both the privacy and accuracy (or utility) of the randomizer. At a high level it is based on replacing the true random bits used to generate the output with pseudo-random bits that can be described using a short seed. For a randomizer $\R \colon X\to Y$, we do this by first picking a fixed reference distribution $\rho$ that is data-independent and $\eps$-close (in the standard sense of differential privacy) to the output distributions of $\R$ for all possible inputs $x\in X$. Existence of such reference distribution is exactly the definition of the deletion version of local differential privacy \citep{ErlingssonFMRSTT2020} and thus our results are easiest to describe in this model.
A sample from $\rho$ typically requires many random bits to generate but, by replacing random bits with pseudo-randomly generated ones, we will obtain a distribution over values in $Y$ that can be described using a short seed. In addition, under standard cryptographic assumptions, a random sample from this distribution is computationally indistinguishable from $\rho$. Given an input $x$ we can now emulate $\R(x)$ by performing rejection sampling relative to pseudo-random samples from $\rho$. A special case of this idea appears in the work of \citet{MishraS} who apply it the problem of estimating sets of counting queries.

A crucial question is whether this scheme satisfies $\eps$ differential privacy. We show that the answer is yes if the pseudo-random generator (PRG) used is strong enough to fool a certain test that looks at the ratio of the output density of $\R(x)$ to $\rho$. This ratio is typically efficiently computable whenever the randomizer itself is efficiently computable. Thus under standard cryptographic assumptions, the privacy is preserved (up to a negligible loss). Similarly, when the processing of the reports on the server side is done by an efficient algorithm the utility will be preserved. See Theorem~\ref{thm:main} for a formal statement. Asymptotically, this result implies that if we assume that there exists an exponentially strong PRG, then the number of bits that needs to be communicated is logarithmic in the running time of the rejection sampler we defined. An immediate practical implication of this result is that in most applications the output of the local randomizer can be compressed to the size of the seed of the system (PRG) without any observable effect on utility or privacy. This size is typically less than 1024 bits. We remark that when implementing a randomizer in practice, true randomness is replaced with pseudo-randomly generated bits with an (implicit) assumption that this does not affect privacy or utility guarantees. Thus the assumptions underlying our analysis are similar to those that are already present in practical implementations of differentially private algorithms.

We demonstrate that this approach also extends to the (more common) replacement notion of local differential privacy and also to $(\eps,\delta)$-DP randomizers. In the latter case the randomizer needs to be modified to allow subsampling via simple truncation. This step adds $\delta$ to both privacy and utility guarantees of the algorithm. For replacement DP this version also requires a more delicate analysis and a stronger set of tests for the PRG. A detailed description of these results is given in Section~\ref{sec:prg}.

An important property of our analysis is that we do not need to follow the general recipe for specific randomizers. Firstly, for some randomizers it is possible to directly sample from the desired distribution over seeds instead of using rejection sampling that requires $e^\eps$ trials (in expectation). In addition, it may be possible to ensure that privacy and utility are preserved without appealing to general cryptographically secure PRGs and associated computational assumptions. In particular, one can leverage a variety of sophisticated results from complexity theory, such as $k$-wise independent PRGs and PRGs for bounded-space computation \citep{Nis92}, to achieve unconditional and more efficient compression.

We apply this fine-grained approach to the problem of frequency estimation over a discrete domain. In this problem the domain $X = [k]$ and the goal is to estimate the frequency of each element $j \in [k]$ in the dataset. This is one of the central and most well-studied problems in  private (federated) data analysis. However, for $\eps > 1$, existing approaches either require communication on the order of $k$ bits, or do not achieve the best known accuracy in some important regimes (see Sec.~\ref{sec:related} for an overview).

The best accuracy is achieved for this problem is achieved by the (asymmetric) RAPPOR algorithm \citep{erlingsson2014rappor} (which has two versions depending on whether it is used with replacement or deletion privacy) and also by the closely related Subset Selection algorithm \citep{wang2016mutual,ye2018optimal}. We observe that a pairwise-independent PRG suffices to fool both the privacy and utility conditions for this randomizer. Thus we can compress RAPPOR to $O(\log k +\eps)$ bits losslessly and unconditionally using a standard construction of a pairwise-independent PRG \citep{luby2006pairwise}. The structure of the PRG also allows us to sample the seeds efficiently without rejection sampling.
The details of this construction appear in Section~\ref{sec:prg}.

As an additional application of our techniques we consider the problem of estimating the mean of $d$-dimensional vectors in $\ell_2$-norm. This problem is a key part of various machine learning algorithms, most notably stochastic gradient descent. In the $\eps > 1$ regime, the first low-communication (specifically, $\lceil \eps \rceil \log_2 d$ bits) and asymptotically optimal algorithm was recently given by~\citet{chen2020breaking}. It is however less accurate empirically and more involved than the algorithm of \citet{bhowmick2019protection} that communicates a $d$ dimensional vector. Using our general result we can losslessly compress the algorithm from \citep{bhowmick2019protection} to $O(\log d + \eps)$ bits. One limitation of this approach is the $O(e^\eps d)$ complexity of rejection sampling in this case which can be prohibitive for large $\eps$. However we show a simple reduction of the $\eps > 1$ case to $\eps < 1$ which increases communication but a factor of $\lceil \eps \rceil$. This general reduction allows us to reduce the running time to $O( \lceil \eps \rceil d)$ and also use a simple and low-communication randomizer that is (asymptotically) optimal only when $\eps < 1$ \citep{duchi2018minimax,ErlingssonFMRSTT2020}. The details of these results and empirical comparisons appear in Section~\ref{sec:mean}.


\subsection{Related Work}
\label{sec:related}
As mentioned, the closest in spirit to our work is the use of rejection sampling in the work of \citet{MishraS}. Their analysis can be seen as a special case of ours but they only prove that the resulting algorithm satisfies $2\eps$-DP. Rejection sampling on a sample from the reference distribution is also used in existing compression schemes \citep{bassily2015local,bun2019heavy} as well as earlier work on private compression in the two-party setting~\citep{McGregorMPRTV}. These approaches assume that the sample is shared between the client and the server, namely, it requires shared randomness. Shared randomness is incompatible with the setting where the report is anonymized and is not directly linked to the user that generated it. As pointed out in \citep{bassily2015local}, a simple way to overcome this problem is to include a seed to a PRG in the output of the randomizer and have the server generate the same sample from the reference distribution as the client. While superficially this approach seems similar to ours, its analysis and properties are different. For example, in our setting only the seed for a single sample that passes rejection sampling is revealed to the server, whereas in  \citep{bassily2015local,bun2019heavy} all samples from the reference distribution are known to the server and privacy analysis does not depend on the strength of the PRG. More importantly, unlike previous approaches our compression scheme is essentially lossless (although at the cost of requiring assumptions for the privacy analysis).

Computational Differential Privacy (CDP) \citep{MironovPRV} is a notion of privacy that defends against computationally bounded adversaries. Our compression algorithm can be easily shown to satisfy the strongest \textsc{Sim}-CDP definition. At the same time,   our privacy bounds also hold for computationally unbounded adversaries as long as the LDP algorithm itself does not lead to a distinguisher. This distinction allows us to remove computational assumptions for specific LDP randomizers.

For both deletion and replacement privacy the best results for frequency estimation are achieved by variants of the RAPPOR algorithm \cite{erlingsson2014rappor} and also by a closely-related Subset Selection algorithm \citep{wang2016mutual,ye2018optimal}.
Unfortunately, both RAPPOR and Subset Selection have very high communication cost of $\approx k H(1/(e^\eps+1))$, where $H$ is the binary entropy function. This has led to numerous and still ongoing efforts to design low-communication protocols for the problem \citep{hsu2012distributed,erlingsson2014rappor,bassily2015local,kairouz2016discrete,wang2016mutual,WangBLJ:17,ye2018optimal,Acharya:2019,acharya2019communication,bun2019heavy,bassily2020practical,chen2020breaking}.

A number of low-communication algorithms that achieve asymptotically optimal bounds in the $\eps < 1$ regime are known \citep{bassily2015local,WangBLJ:17,Acharya:2019,acharya2019communication,bassily2020practical,chen2020breaking}. The first low-communication algorithm that achieves asymptotically optimal bounds in the $\eps > 1$ regime is given in \citep{WangBLJ:17}. It communicates $O(\eps)$ bits and relies on shared randomness. However, it matches the bounds achieved by RAPPOR only when $e^\eps$ is an integer.
\citet{acharya2019communication} and \citet{chen2020breaking} give closely related approaches that are asymptotically optimal and use $\log_2 k$ bits of communication (without shared randomness). However both the theoretical bounds and empirical results for these algorithms are noticeably worse than those of (asymmetric) RAPPOR and Subset Selection (e.g.~plots in \citep{chen2020breaking} show that these algorithms are $\approx$15-20\% worse for $\eps =5$ than Subset Selection\footnote{The error of asymmetric RAPPOR (namely 0 and 1 are flipped with different probabilities) is essentially identical to that of the Subset Selection randomizer. Comparisons with RAPPOR often use the symmetric RAPPOR  which is  substantially worse than the asymmetric version for the replacement notion of differential privacy. See Section~\ref{sec:rappor} for details.}). The constructions in \citep{acharya2019communication,chen2020breaking} and their analysis are also substantially more involved than RAPPOR.

A closely related problem is finding ``heavy hitters'', namely all elements $j \in [k]$ with counts higher than some given threshold. In this problem the goal is to avoid linear runtime dependence on $k$ that would result from doing frequency estimation and then checking all the estimates. This problem is typically solved using a ``frequency oracle'' which is an algorithm that for a given $j\in [k]$ returns an estimate of the number of $j$'s held by users (typically without computing the entire histogram) \citep{bassily2015local,bassily2020practical,bun2019heavy}. Frequency estimation is also closely related to the discrete distribution estimation problem in which inputs are sampled from some distribution over $[k]$ and the goal is to estimate the distribution \citep{ye2018optimal,Acharya:2019,acharya2019communication}. Indeed, bounds for frequency estimation can be translated directly to bounds on distribution estimation by adding the sampling error.

Mean estimation has attracted a lot of attention in recent years as it is an important subroutine in differentially private (stochastic) gradient descent algorithms~\cite{BassilyST14a, AbadiCGMMTZ:16} used in private federated learning~\cite{kairouz2019advances}.
Indeed, private federated optimization algorithms aggregate updates to the model coming from each client in a batch of clients by getting a private estimate of the average update. When the models are large, the dimensionality of the update $d$ leads to significant communication cost. Thus reducing the communication cost of mean estimation has been studied in many works with~\cite{AgarwalSYKM18, girgis2020shuffled, chen2020breaking, GKMM19} or without privacy~\citep{AlistrahGLTV17,FaghriTMARR20,SYKM17,GKMM19,MT20}.

In the absence of communication constraints and $\eps < d$, the optimal $\eps$-LDP protocols for this problem achieve an expected squared $\ell_2$ error of $\Theta(\frac{d}{n \min(\eps, \eps^2)})$ \citep{duchi2018minimax,DuchiR19}.  When $\eps \leq 1$, the randomizer of~\citet{duchi2018minimax} also achieves the optimal $O(\frac{d}{n\eps^2})$ bound. Recent work of~\citet{ErlingssonFMRSTT2020}  gives a low-communication version of this algorithm. Building on the approach in~\citep{duchi2018minimax}, \citet{bhowmick2019protection} describe the $\PrivUnit$ algorithm that achieves the asymptotically optimal accuracy also when $\eps > 1$ but has communication cost of $\Omega(d)$.

An alternative approach in the $\eps < 1$ regime was given by~\citet{FeldmanGV:15} who show that the mean estimation problem can be solved by having each client answer a single counting query. This approach is based on Kashin's representation that maps vectors in the unit $d$-dimensional ball to vectors in $[-1,1]^{O(d)}$ ~\citep{Lyubarskii:2010}. Their work does not explicitly discuss the communication cost and assumes that the server can pick the randomizer used at each client. However it is easy to see that a single bit suffices to answer a counting query and therefore an equivalent randomizer can be implemented using $\lceil \log_2 d \rceil +1$ bits of communication (or just 1 bit if shared randomness is used). \citet{chen2020breaking} give a randomizer based on the same idea that also achieves the asymptotically optimal bound in the $d > \eps > 1$ regime. Their approach uses $\lceil \eps \rceil \log_2 d$ bits of communication. Computing Kashin's representation is more involved than algorithms in \citep{duchi2018minimax,bhowmick2019protection}. In addition, as we demonstrate empirically\footnote{Plots in \citep{chen2020breaking} also compare their algorithm with $\PrivUnit$ yet as their code at \citep{KashinImplement2021} shows and was confirmed by the authors, they implemented the algorithm from \citep{duchi2018minimax} instead of $\PrivUnit$ which is much worse than $\PrivUnit$ for $\eps =5$. The authors also confirmed that parameters stated in their figures are incorrect so cannot be directly compared to our results.}, the variance of the estimate resulting from this approach is nearly a factor of $5\times$ larger for typical parameters of interest.


\eat{
Dart throwing goes back to at least~\citet{Broder}, and has been used in complexity theory in~\citet{Holenstein}.
Unfortunately, both RAPPOR and Subset Selection have very high communication cost of $\approx k H(1/(e^\eps+1))$, where $H$ is the binary entropy function. This has lead to numerous and still ongoing efforts to design low-communication protocols for the problem \citep{hsu2012distributed,erlingsson2014rappor,bassily2015local,kairouz2016discrete,wang2016mutual,WangBLJ:17,ye2018optimal,Acharya:2019,acharya2019communication,bun2019heavy,bassily2020practical,chen2020breaking}.
}

\section{Preliminaries}
For a positive integer $k$ we denote $[k] = \{1,2\ldots,k\}$. For an arbitrary set $S$ we use $x\sim S$ to mean that $x$ is chosen randomly and uniformly from $S$.

Differential privacy (DP) is a measure of stability of a randomized algorithm. It bounds the change in the distribution on the outputs when one of the inputs is either removed or replaced with an arbitrary other element. The most common way to measure the change in the output distribution is via approximate infinity divergence. More formally, we say that two probability distributions $\mu$ and $\nu$ over (finite) domain $Y$ are $(\eps, \delta)$-\emph{close} if for all $E \subset Y$, \[e^{-\eps} (\mu(E) - \delta) \le \nu(E) \le e^{\eps} \mu(E) +\delta.\]
This condition is equivalent to $\sum_{y\in Y} |\mu(y) - e^\eps \nu(y)|_+ \leq \delta$ and $\sum_{y\in Y} |\nu(y) - e^\eps \mu(y)|_+ \leq \delta$, where $|a|_+ \dfn \max\{a,0\}$ \citep{DworkRoth:14}.
We also say that two random variables $P$ and $Q$ are $(\eps, \delta)$-close if their probability distributions are $(\eps,\delta)$-close. We abbreviate $(\eps,0)$-close to $\eps$-close.


Algorithms in the local model of differential privacy and federated data analysis rely on the notion of {\em local randomizer}.
\begin{defn}
An algorithm $\R\colon X\to Y$ is an $(\eps, \delta)$-DP \emph{local randomizer}  if
for all pairs $x,x'\in \D$, $\R(x)$ and $\R(x')$ are $(\eps, \delta)$-close.
\end{defn}
We will also use the add/delete variant of differential privacy which was defined for local randomizers in \citep{ErlingssonFMRSTT2020}.
\begin{defn}
An algorithm $\R\colon X\to Y$ is a {\em deletion $(\eps, \delta)$-DP local randomizer} if
there exists a reference distribution $\reference$ such that for all data points $x\in X$, $\R(x)$ and $\reference$ are $(\eps, \delta)$-close.
\end{defn}
It is easy to see that a replacement $(\eps,\delta)$-DP algorithm is also a deletion $(\eps,\delta)$-DP algorithm, and that
a deletion $(\eps,\delta)$-DP algorithm is also a replacement $(2\eps,2\delta)$-DP algorithm.

\paragraph{Fooling and Pseudorandomness:}
The notion of pseudorandomness relies on ability to distinguish between the output of the generator and true randomness using a family of tests, where a test is a boolean function (or algorithm).

\begin{defn}
Let $\D$ be a family of boolean functions over some domain $Y$. We say that two random variables $P$ and $Q$ over $Y$ are $(\D,\beta)$-indistinguishable if for all $D\in \D$,
\[\left| \pr[D(P) =1] - \pr[D(Q) =1] \right| \leq \beta .\]
We say that $P$ and $Q$ are $(T,\beta)$-computationally indistinguishable if $P$ and $Q$ are $(\D,\beta)$-indistinguishable with $\D$ being all tests that can be computed in time $T$ (for some fixed computational model such as boolean circuits).
\end{defn}

We now give a definition of a pseudo-random number generator.
\begin{defn}[Pseudo-random generator]\label{def:prg}
We say that an algorithm $G\colon \zo^n \to \zo^m$ where $m \gg n$, $\beta$-fools a family of tests $\D$ if $G(s)$ for $s \sim \zo^n$ is $(\D,\beta)$-indistinguishable from $r$ for $r\sim \zo^m$. We refer to such an algorithm as $(\D,\beta)$-PRG and also use $(T,\beta)$-PRG to refer to $G$ that $\beta$-fools all tests running in time $T$.
\end{defn}

Standard cryptographic assumptions (namely that one-way functions exist) imply that for any $m$ and $T$ that are polynomial in $n$ there exists an efficiently computable $(\beta,T)$-PRG $G$, for negligible $\beta$ (namely, $\beta = 1/n^{\omega(1)}$). For a number of standard approaches to cryptographically-secure PRGs, no tests are known that can  distinguish the output of the PRG from true randomness with $\beta = 2^{-o(n)}$ in time $T= 2^{o(n)}$. For example finding such a test for a PRG based on SHA-3 would be a major breakthrough. To make the assumption that such a test does not exist we refer to a $(\beta,T)$-PRG for $\beta = 2^{-\Omega(n)}$ and $T= 2^{\Omega(n)}$ as an {\em exponentially strong} PRG.

\eat{
The simulation-based notion of computational differential privacy (CDP) is defined as computational indistinguishability from an algorithm that is differentially private. This notion was defined in the context of algorithms operating on a dataset.

\begin{defn}[Simulation-based computational DP]\label{def:cdp}
An algorithm $\A \colon X^n \to Y$ is $(T,\beta)$-computationally $(\eps,\delta)$-DP if there exists an $(\eps,\delta)$-DP $\A' \colon X^n \to Y$ such that for all datasets $S \in X^n$, $\A(S)$ is $(T,\beta)$-computationally indistinguishable from $A'(S)$.
\end{defn}
}

\section{Local Pseudo-Randomizers}
\label{sec:prg}
In this section we describe a general way to compress LDP randomizers that relies on the complexity of the randomizer and subsequent processing. We will first describe the result for deletion $\eps$-DP and then give the versions for replacement DP and $(\eps,\delta)$-DP.

For the purpose of this result we first need to quantify how much randomness a local randomizer needs. We will say that a randomizer $\R\colon X\to Y$ is $t$-samplable if there exists a deterministic algorithm $\R_\emptyset \colon \zo^t \to Y$ such that for $r$ chosen randomly and uniformly from $\zo^t$, $\R_\emptyset(r)$ is distributed according to the reference distribution of $\R$ (denoted by $\reference$).
Typically, for efficiently computable randomizers, $t$ is polynomial in the size of the output $\log(|Y|)$ and $\eps$. Note that every value $y$ in the support of $\rho$ is equal to $\R_\emptyset(r)$ for some $r$. Thus every element that can be output by $\R$ can be represented by some $r\in \zo^t$.

Our goal is to compress the communication by restricting the output from all those that can be represented by $r\in  \zo^t$ to all those values in $Y$ that can be represented by a $t$-bit string generated from a seed of length $\ell \ll t$ using some PRG $G\colon \zo^\ell \to \zo^t$. We could then send the seed to the server and let the server first generate the full $t$-bit string using $G$ and then run $\R_\emptyset$ on it. The challenge is to do this efficiently while preserving the privacy and utility guarantees of $\R$.

Our approach is based on the fact that we can easily sample from the pseudo-random version of the reference distribution $\rho$ by outputting  $\R_\emptyset(G(s))$ for a random and uniform seed $s$. This leads to a natural way to define a distribution over seeds on a input $x$: a seed $s$ is output with probability that is proportional to $\frac{\pr[\R(x) =\R_\emptyset(G(s))]}{\pr_{r\sim \zo^t}[\R_\emptyset(r) = \R_\emptyset(G(s))]}$. Specifically we define the desired randomizer as follows.
\begin{defn}
\label{def:rg}
For a $t$-samplable deletion DP local randomizer $\R\colon X\to Y$ and a function $G \colon \zo^\ell \to \zo^t$ let $\R[G]$ denote the local randomizer that given $x\in X$, outputs $s\in \zo^t$ with probability proportional to $\frac{\pr[\R(x) = \R_\emptyset(G(s))]}{\pr_{r\sim \zo^t}[\R_\emptyset(r) = \R_\emptyset(G(s))]}$.
\end{defn}

For some combinations of a randomizer $\R$ and PRG $G$ there is an efficient way to implement $\R[G]$ directly (as we show in one of our applications). In the general case, when such algorithm may not exist we can sample from $R[G](x)$ by applying rejection sampling to uniformly generated seeds. A special case of this approach is implicit in the work of \citet{MishraS} (albeit with a weaker analysis). Rejection sampling only requires an efficient algorithm for computing the ratio of densities above to sufficiently high accuracy. We describe the resulting algorithm below.

\begin{algorithm}[htb]
	\caption{$\R[G,\gamma]$: PRG compression of $\R$}\label{alg:compressv1}
	\begin{algorithmic}[1]
		\REQUIRE $x\in X$, $\eps,\gamma >0$; seeded PRG $G \colon \zo^\ell \to \zo^t$; $t$-samplable $\eps$-DP randomizer $\R$.
        \STATE $J = e^\eps \ln(1/\gamma)$
        \FOR {$j=1,\ldots,J$}
		\STATE Sample a random seed $s \in \zo^\ell$.
		    \STATE $y = \R_\emptyset(G(s))$
		    \STATE Sample $b$ from  $\Bern\left(\frac{\pr[\R(x) = y]}{e^\eps \pr_{r\sim \zo^t}[\R_\emptyset(r) = y]}\right)$
            \IF{$b == 1$}
                \STATE BREAK
            \ENDIF
		\ENDFOR
		\STATE Send $s$
\end{algorithmic}
\end{algorithm}

Naturally, the output of this randomizer can be decompressed by applying $G\circ \R_\emptyset$ to it. It is also clear that the communication cost of the algorithm is $\ell$ bits.

Next we describe the general condition on the PRG $G$ that suffices for ensuring that the algorithm that outputs a random seed with correct probability is differentially private.
\begin{lem}
\label{lem:prg-privacy}
For  a $t$-samplable deletion $\eps$-DP local randomizer $\R\colon X\to Y$ and $G \colon \zo^\ell \to \zo^t$, let $\D$ denote the following family of tests which take $r'\in \zo^t$ as an input:
\ifconf
\else
 \[ \D \dfn \left\{\left. \ind{\frac{\pr[\R(x) = \R_\emptyset(r')]}{\pr_{r\sim \zo^t}[\R_\emptyset(r) = \R_\emptyset(r')]} \geq \theta}\ \right|\ x\in X,\theta \in [0,e^\eps]\right\} , \]
 \fi
 where $\mathtt{ind}$ denotes the $\zo$ indicator function of a condition.
 If $G$ $\beta$-fools $\D$ for $\beta < 1/(2e^\eps)$ then $R[G]$ is a deletion $(\eps+ 2 e^\eps \beta)$-DP local randomizer. Furthermore, for every $\gamma > 0$, $\R[G,\gamma]$ is a deletion $(\eps+ 2 e^\eps \beta)$-DP local randomizer.
\end{lem}
\ifconf\else
\begin{proof}
We demonstrate that if $\R[G]$ is not a deletion $(\eps+ 2 e^\eps \beta)$-DP randomizer then there exists a test in $\D$ that distinguishes the output of $G$ from true randomness that succeeds with probability at least $\beta$.
To analyze the privacy guarantees of $\R[G]$ we let the reference distribution $\rho_G$ be the uniform distribution over $\zo^\ell$.
For brevity, for $y\in Y$ we denote the density ratio of $\R(x)$ to $\rho$ at $y$ by
\[ \pi_x(y) \dfn \frac{\pr[\R(x) = y]}{\pr_{r\sim \zo^t}[\R_\emptyset(r) = y]} .\]
Then, $\R[G]$ outputs a seed $s$ with probability:
\[ \mu_x(s) \dfn \frac{\pi_x(\R_\emptyset(G(s)))}{\sum_{s' \in \zo^\ell}\pi_x(\R_\emptyset(G(s')))} .\]
By definition of our reference distribution, $\rho_G(s) = 2^{-\ell}$ for all $s$. Therefore
\[\frac{\mu_x(s)}{\rho_G(s)} = \frac{ \mu_x(s)}{2^{-\ell}} =  \frac{\pi_x(\R_\emptyset(G(s)))}{\E_{s' \sim \zo^\ell}[\pi_x(\R_\emptyset(G(s')))]} .\]

We observe that, by the fact that $\R(x)$ is $\eps$-DP we have that $\pi_x(\R_\emptyset(G(s))) \in [e^{-\eps},e^\eps]$. Therefore, to show that $\R[G]$ is $(\eps +2 e^\eps \beta)$-DP, it suffices to show that the denominator is in the range $[e^{-2 e^\eps \beta},e^{2 e^\eps \beta}]$. To show this, we assume for the sake of contradiction that it is not true. Namely, that either
\[ \E_{s' \sim \zo^\ell}[\pi_x(\R_\emptyset(G(s')))]  > e^{2 e^\eps \beta} > 1+ e^\eps \beta \] or
\[ \E_{s' \sim \zo^\ell}[\pi_x(\R_\emptyset(G(s')))] < e^{-2 e^\eps \beta} <1- e^\eps \beta, \]
where we used the assumption that $ \beta <1/(2 e^\eps)$ in the second inequality.

We first deal with the case when $\E_{s' \sim \zo^\ell}[\pi_x(\R_\emptyset(G(s')))] > 1+e^\eps \beta$ (as the other case will be essentially identical). Observe that for true randomness we have:
\[\E_{r' \sim \zo^t} \lb\pi_x(\R_\emptyset(r')) \rb =  \E_{r' \sim \zo^t} \lb \frac{\pr[\R(x) = \R_\emptyset(r')]}{\pr_{r\sim \zo^t}[\R_\emptyset(r) = \R_\emptyset(r')]}\rb = 1 .\]
Using the fact that $\pi_x(y) \in [0,e^\eps]$ we have that
\[ \E_{s' \sim \zo^\ell}[\pi_x(\R_\emptyset(G(s')))] = \int_{0}^{e^\eps} \pr_{s' \sim \zo^\ell}[\pi_x(\R_\emptyset(G(s'))) \geq \theta] d\theta \] and, similarly,
\[\E_{r' \sim \zo^t} \lb\pi_x(\R_\emptyset(r')) \rb = \int_{0}^{e^\eps} \pr_{r' \sim \zo^t} \lb\pi_x(\R_\emptyset(r')) \geq \theta \rb d\theta .\]
Thus, by our assumption,
\[\int_{0}^{e^\eps} \lp \pr_{s' \sim \zo^\ell}[\pi_x(\R_\emptyset(G(s'))) \geq \theta] - \pr_{r' \sim \zo^t} \lb\pi_x(\R_\emptyset(r')) \geq \theta \rb \rp d\theta > e^\eps \beta .\]
This implies that there exists $\theta \in [0,e^\eps]$ such that
\[\pr_{s' \sim \zo^\ell}[\pi_x(\R_\emptyset(G(s'))) \geq \theta] - \pr_{r' \sim \zo^t} \lb\pi_x(\R_\emptyset(r')) \geq \theta \rb > \beta .\]
Note that $\ind{\pi_x(\R_\emptyset(r')) \geq \theta} \in \D$, for all $x \in X$ and $\theta \in [0,e^\eps]$ contradicting our assumption on $G$. Thus we obtain that $\E_{s' \sim \zo^\ell}[\pi_x(\R_\emptyset(G(s')))] \leq 1+e^\eps \beta$. We can arrive at a contradiction in the case when $\E_{s' \sim \zo^\ell}[\pi_x(\R_\emptyset(G(s')))] < 1 - e^\eps \beta$ in exactly the same way.

To show that $\R[G,\gamma]$ is a deletion $(\eps+ 2 e^\eps \beta)$-DP local randomizer we observe that for every $x$, conditioned on accepting one of the samples, $\R[G,\gamma](x)$ outputs a sample distributed exactly according to $\R[G](x)$. If $\R[G,\gamma](x)$ does not accept any samples than it samples from the reference distribution $\rho_G$. Thus given that $\R[G](x)$ is $(\eps+ 2 e^\eps \beta)$-close to $\rho_G$ we have that the output distribution $\R[G,\gamma](x)$ is also $(\eps+ 2 e^\eps \beta)$-close to $\rho_G$.
\end{proof}
\fi
Unlike the preservation of privacy, conditions on the PRG under which we can ensure that the utility of $\R$ is preserved depend on the application. Here we describe a general result that relies only on the efficiency of the randomizer to establish computational indistinguishability of the output of our compressed randomizer from the output of the original one.


\ifconf\else
As the first step, we show that when used with the identity $G$, the resulting randomizer is $\gamma$-close in total variation distance to $\R$.
\begin{lem}
\label{lem:rejection-randomizer}
Let $\R$ be a deletion $\eps$-DP $t$-samplable local randomizer. Then for the identity function $\ID_t \colon \zo^t \to \zo^t$ and any $\gamma > 0$ we have that $\R[\ID_t,\gamma]$ is a deletion $\eps$-DP local randomizer and for every $x\in \X$, $\TV(\R_\emptyset(\R[\ID_t,\gamma](x)),\R(x)) \leq \gamma$.
\end{lem}
\ifconf\else
\begin{proof}
When applied with $G= \ID_t$, $y$ is distributed according to the reference distribution of $\R$. Thus
the algorithm performs standard rejection sampling until it accepts a sample or exceeds the bound $J$ on the number of steps. Note that deletion DP implies that $\frac{\pr[\R(x) = y]}{e^\eps \pr_{r\sim \zo^t}[\R_\emptyset(r) = y]} \leq 1$.  At each step, conditioned on success the algorithm samples $s$ such that $\R_\emptyset(s)$ is distributed identically to $\R(x)$.  Further, the acceptance probability at each step is
 \[ \E_{y \sim \rho} \lb \frac{\pr[\R(x) = y]}{e^\eps \pr_{r\sim \zo^t}[\R_\emptyset(r) = y]} \rb =  \sum_{y\in Y} \frac{\pr[\R(x) = y]}{e^\eps}  = \frac{1}{e^\eps} .\]
Thus the probability that all steps reject is $\leq (1- e^{-\eps})^J \leq \gamma$. It follows that $\TV(\R_\emptyset(\R[\ID_t,\gamma](x)),\R(x))$ is bounded by  $\gamma$ as claimed.
\end{proof}
\fi
We can now state the implications of using a sufficiently strong PRG on the output of the randomizer.
\fi
\begin{lem}
\label{lem:prg-compute}
Let $\R$ be a deletion $\eps$-DP $t$-samplable local randomizer, let $G \colon \zo^\ell \to \zo^t$ be $(T,\beta)$-PRG.
Let $T(\R,G,\gamma)$ denote the running time of $\R[G,\gamma]$ and assume that $T > T(\R,G,\gamma)$.
Then for all $x \in X$, $\R_\emptyset(G(\R[G,\gamma](x)))$ is $(T',\beta')$-computationally indistinguishable from $\R(x)$, where $\beta' = \gamma + e^\eps \ln(1/\gamma) \beta$ and $T' = T - T(\R,G,\gamma)$.
\end{lem}
\ifconf\else \begin{proof}
By Lemma~\ref{lem:rejection-randomizer}, $\TV(\R_\emptyset(\R[\ID_t,\gamma](x)),\R(x)) \leq \gamma$ and thus it suffices to prove that $\R_\emptyset(G(\R[G,\gamma](x)))$ is $(T', e^\eps \ln(1/\gamma) \beta)$-computationally indistinguishable from $\R_\emptyset(\R[\ID_t,\gamma](x))$.
Towards a contradiction, suppose that there exists a test $D'$ running in time $T'$ such that for some $x$,
\begin{align*}
   \big|  \pr[&D'(\R_\emptyset(G(\R[G,\gamma](x)))) =1] -  \\  & \pr[D'(\R_\emptyset(\R[\ID_t,\gamma](x))) =1] \big| \geq e^\eps \ln(1/\gamma) \beta .
\end{align*}
We claim that there exists a test for distinguishing $G(s)$ for $s \sim \zo^\ell$ from a truly random seed $r \sim \zo^t$. Note that
$\R_\emptyset(G(\R[G,\gamma]))$ can be seen as $\R[G,\gamma]$ that outputs directly $y=\R_\emptyset(G(s))$ instead of $s$ itself. Next we observe that $\R_\emptyset(G(\R[G,\gamma]))$ uses the output of $G$ at most $J= e^\eps \ln(1/\gamma)$ times in place of truly random $t$-bit string used by $\R_\emptyset(\R[\ID_t,\gamma])$. Thus, by the standard hybrid argument, one of those applications can be used to test $G$ with success probability at least $e^\eps \ln(1/\gamma) \beta/J = \beta$. This test requires running a hybrid between $\R_\emptyset(G(\R[G,\gamma]))$ and $\R_\emptyset(\R[\ID_t,\gamma])$ in addition to $D'$ itself. The resulting test runs in time $T'+T(\R, G, \gamma) = T$.
Thus we obtain a contradiction to $G$ being a $(T,\beta)$-PRG.
\end{proof}\fi

As a direct corollary of Lemmas~\ref{lem:prg-privacy} and \ref{lem:prg-compute} we obtain a general way to compress efficient LDP randomizers.
\begin{thm}
\label{thm:main}
Let $\R$ be a deletion $\eps$-DP $t$-samplable local randomizer, let $G \colon \zo^\ell \to \zo^t$ be $(T,\beta)$-PRG for $\beta < 1/(2e^\eps)$. Let $T(\R,G,\gamma)$ be the running time of $\R[G,\gamma]$ and assume that $T > T(\R,G,\gamma)$.
Then  $\R[G,\gamma]$ is a deletion $(\eps+ 2 e^\eps \beta)$-DP local randomizer and for all $x \in X$, $\R_\emptyset(G(\R[G,\gamma](x)))$ is $(T',\beta')$-computationally indistinguishable from $\R(x)$, where $\beta' = \gamma + e^\eps \ln(1/\gamma) \beta$ and $T' = T - T(\R,G,\gamma)$.
\end{thm}
\ifconf\else \begin{proof}
The second part of the claim is exactly Lemma~\ref{lem:prg-compute}. To see the first part of the claim note that by our assumption $T > T(\R,G,\gamma)$ and computation of the ratio of densities $\frac{\pr[\R(x) = \R_\emptyset(r')]}{e^\eps \pr_{r\sim \zo^t}[\R_\emptyset(r) = \R_\emptyset(r')]}$ for any $r' \in \zo^t$ is part of $\R[G,\gamma]$. This implies that the test family $\D$ defined in Lemma~\ref{lem:prg-privacy} can be computed in time $T$. Now applying Lemma~\ref{lem:prg-privacy} gives us the privacy claim.
\end{proof}\fi

By plugging an exponentially strong PRG $G$ into Theorem~\ref{thm:main} we obtain that if an LDP protocol based on $\R$ runs in time $T$ then its communication can be compressed to $O(\log(T+T(\R,G,\gamma))$ with negligible effect on privacy and utility. We also remark that even without making any assumptions on $G$, $\R[G,\gamma]$ satisfies $2\eps$-DP. In other words, failure of the PRG does not lead to a significant privacy violation, beyond the degradation of the privacy parameter $\eps$ by a factor of two.
\begin{lem}
\label{lem:basic-privacy}
Let $\R$ be a deletion $\eps$-DP $t$-samplable local randomizer, let $G \colon \zo^\ell \to \zo^t$ be an arbitrary function.
Then  $\R[G,\gamma]$ is a deletion $2\eps$-DP local randomizer.
\end{lem}
\ifconf\else \begin{proof}
As in the proof of Lemma~\ref{lem:prg-privacy} we observe that if we take the reference distribution to be uniform over $\zo^\ell$ we will get that, conditioned on accepting a sample, the seed $s$ is output with probability $\mu_x(s)$ such that
\[\frac{\mu_x(s)}{\rho_G(s)} = \frac{ \mu_x(s)}{2^{-\ell}} =  \frac{\pi_x(\R_\emptyset(G(s)))}{\E_{s' \sim \zo^\ell}[\pi_x(\R_\emptyset(G(s')))]} .\]
By the fact that $\R(x)$ is $\eps$-DP we have that for every $s' \in \zo^\ell$, $\pi_x(\R_\emptyset(G(s'))) \in [e^{-\eps},e^{\eps}]$ and thus
$\frac{\mu_x(s)}{\rho_G(s)} \in [e^{-2\eps},e^{2\eps}]$.
\end{proof}\fi

\ifconf
\else
\subsection{Replacement DP}
We now show that the same approach can be used to compress a replacement $\eps_r$-DP randomizer $\R$. To do this we first let 
$\rho$ be some reference distribution relative to which $\R$ is deletion $\eps$-DP for some $\eps \leq \eps_r$. One possible way to define $\rho$ is to pick some fixed $x_0 \in X$ and let $\rho$ be the distribution of $\R(x_0)$. In this case $\eps = \eps_r$. But other choices of $\rho$ are possible that give an easy to sample distribution and $\eps < \eps_r$. In fact, for some standard randomizers such as addition of Laplace noise we will get $\eps = \eps_r/2$.

Now assuming that $\rho$ is $t$-samplable and given a PRG $G\colon \zo^\ell \to \zo^t$ we define $\R[G]$ as in Def.~\ref{def:rg} and $\R[G,\gamma]$ us in Algorithm~\ref{alg:compressv1}. The randomizer $\R$ is deletion $\eps$-DP so all the results we proved apply to it as well (with the deletion $\eps$ and not the replacement $\eps_r$). In addition we show that replacement privacy is preserved as well.

\begin{lem}
\label{lem:prg-privacy-replace}
For  a $t$-samplable deletion $\eps$-DP and replacement $\eps_r$-DP local randomizer $\R\colon X\to Y$ and $G \colon \zo^\ell \to \zo^t$, let $\D$ denote the following family of tests which take $r'\in \zo^t$ as an input:
 \[ \D \dfn \left\{\left. \ind{\frac{\pr[\R(x) = \R_\emptyset(r')]}{\pr_{r\sim \zo^t}[\R_\emptyset(r) = \R_\emptyset(r')]} \geq \theta}\ \right|\ x\in X,\theta \in [0,e^\eps]\right\}. \]
 If $G$ $\beta$-fools $\D$ for $\beta < 1/(2e^\eps)$ then $R[G]$ is a replacement $(\eps_r + 4 e^\eps \beta)$-DP local randomizer. Furthermore, for every $\gamma > 0$, $\R[G,\gamma]$ is a replacement $(\eps_r+ 4 e^\eps \beta)$-DP local randomizer.
\end{lem}
\begin{proof}
As in the proof of Lemma~\ref{lem:prg-privacy}, for $y\in Y$, we denote the density ratio of $\R(x)$ to $\rho$ at $y$ by
\[ \pi_x(y) \dfn \frac{\pr[\R(x) = y]}{\pr_{r\sim \zo^t}[\R_\emptyset(r) = y]} \]
and note that $\R[G]$ outputs a seed $s$ with probability:
\[ \mu_x(s) \dfn \frac{\pi(\R_\emptyset(G(s)))}{\sum_{s' \in \zo^\ell}\pi(\R_\emptyset(G(s')))} .\]
Thus for two inputs $x,x' \in X$ and any $s \in \zo^\ell$ we have that
\alequn{ \frac{\mu_x(s)}{\mu_{x'}(s)} &=  \frac{\pi_x(\R_\emptyset(G(s)))}{\pi_{x'}(\R_\emptyset(G(s)))} \cdot \frac{\sum_{s' \in \zo^\ell}\pi_{x'}(\R_\emptyset(G(s')))}{\sum_{s' \in \zo^\ell}\pi_x(\R_\emptyset(G(s')))} \\
& = \frac{\pr[\R(x) = \R_\emptyset(G(s'))]}{\pr[\R(x') = \R_\emptyset(G(s'))]} \cdot \frac{\E_{s' \sim \zo^\ell}[\pi_{x'}(\R_\emptyset(G(s')))]}{\E_{s' \sim \zo^\ell}[\pi_x(\R_\emptyset(G(s')))]}
}
Now $\R$ is $\eps_r$-replacement-DP and therefore the first term satisfies:
\[ \frac{\pr[\R(x) = \R_\emptyset(G(s'))]}{\pr[\R(x') = \R_\emptyset(G(s'))]} \in \lb e^{-\eps_r}, e^{\eps_r} \rb. \]
At the same time, we showed in Lemma~\ref{lem:prg-privacy} that $\E_{s' \sim \zo^\ell}[\pi_x(\R_\emptyset(G(s')))] \in [e^{-2 e^\eps \beta},e^{2 e^\eps \beta}]$ and also $\E_{s' \sim \zo^\ell}[\pi_{x'}(\R_\emptyset(G(s')))] \in [e^{-2 e^\eps \beta},e^{2 e^\eps \beta}]$.
Therefore $ \frac{\mu_x(s)}{\mu_{x'}(s)} \in [e^{-\eps_r-4 e^\eps \beta},e^{\eps_r + 4 e^\eps \beta}]$.

To show that $\R[G,\gamma]$ is a replacement $(\eps_r+ 4 e^\eps \beta)$-DP local randomizer we observe that for every $x$, $\R[G,\gamma](x)$ is a mixture of $\R[G](x)$ and $\rho_G$. As we showed, $\R[G](x)$ is $(\eps_r+ 4 e^\eps \beta)$-close to $\R[G](x')$ and we also know from Lemma~\ref{lem:prg-privacy} that $\rho_G$ is $(\eps+ 2 e^\eps \beta)$-close to $\R[G](x')$. By quasi-convexity we obtain that $\R[G,\gamma](x)$ is $(\eps_r+ 4 e^\eps \beta)$-close to $\R[G](x')$. We also know that $\R[G,\gamma](x)$ is $(\eps+ 2 e^\eps \beta)$-close to $\rho_G$. Appealing to quasi-convexity again, we obtain that $\R[G,\gamma](x)$ is $(\eps_r+ 4 e^\eps \beta)$-close to $\R[G,\gamma](x')$.
\end{proof}

\subsection{Extension to $(\eps,\delta)$-DP}
We next extend our approach to $(\eps, \delta)$-DP randomizers. The approach here is similar, except that we for some outputs $y=\R_\emptyset(G(s))$, the prescribed ``rejection probability'' in the original approach would be larger than one. To handle this, we simply truncate this ratio at $1$ to get a probability. Algorithm~\ref{alg:compressv2} is identical to Algorithm~\ref{alg:compressv1} except for this truncation in the step where we sample $b$.

\begin{algorithm}[htb]
	\caption{$\R[G,\gamma]$: PRG compression of deletion $(\eps,\delta)$-DP $\R$}\label{alg:compressv2}
	\begin{algorithmic}[1]
		\REQUIRE $x\in X$, $\eps,\gamma >0$; seeded PRG $G \colon \zo^\ell \to \zo^t$; $t$-samplable $\eps$-DP randomizer $\R$.
        \STATE $J= e^\eps \ln(1/\gamma)/(1-\delta)$
        \FOR {$j=1, \ldots, J$}
     		\STATE Sample a random seed $s \in \zo^\ell$.
		    \STATE $y = \R_\emptyset(G(s))$
		    \STATE Sample $b$ from $\Bern\lp\min \lp 1, \frac{\pr[\R(x) = y]}{e^\eps \pr_{r\sim \zo^t}[\R_\emptyset(r) = y]}\rp\rp$
            \IF{$b == 1$}
                \STATE BREAK
            \ENDIF
		\ENDFOR
		\STATE Send $s$
\end{algorithmic}
\end{algorithm}

The proof is fairly similar to that for the pure DP randomizer. We start with a lemma that relates the properties of the PRG to the properties of the randomizer that need to be preserved in order to ensure that it satisfies deletion $(\eps',\delta')$-LDP.
\begin{lem}
\label{lem:prg-ed-deletion}
For  a $t$-samplable deletion $(\eps, \delta)$-DP local randomizer $\R\colon X\to Y$ and $G \colon \zo^\ell \to \zo^t$, let $\D$ denote the following family of tests which take $r'\in \zo^t$ as an input:
 \[ \D \dfn \left\{\left. \ind{\frac{\pr[\R(x) = \R_\emptyset(r')]}{\pr_{r\sim \zo^t}[\R_\emptyset(r) = \R_\emptyset(r')]} \geq \theta}\ \right|\ x\in X,\theta \in [0,e^\eps]\right\} . \]
Suppose that $G$ $\beta$-fools $\D$ and let $ \pi_x(y) \dfn \frac{\min(\pr[\R(x) = y], e^{\eps}\pr[\R(\emptyset) = y])}{\pr_{r\sim \zo^t}[\R_\emptyset(r) = y]}$.
Then \[\E_{s' \sim \zo^\ell}[\pi_x(\R_\emptyset(G(s')))] \in [1-\delta-e^\eps\beta, 1+e^\eps\beta]\] and 
\[\E_{s' \sim \zo^\ell} \lb \left| 1- e^\eps \pi_x(\R_\emptyset(G(s'))) \right|_+ \rb \leq \delta + \beta .\]
\end{lem}
\begin{proof}
Let $\rho_G$ be the uniform distribution over $\zo^\ell$. Let $\nu_x(y) \dfn \pr[\R(x)=y]$ and let $\nut_x(y) \dfn \min(\nu_x(y), e^{\eps}\pr[\R(\emptyset) = y])$. Note that $\nut_x(\cdot)$ does not necessarily define a probability distribution. For $S=\{y: \nut_x(y) < \nu_x(y)\}$, we have
\begin{align*}
    \nu_x(S) &= \sum_{y \in S} \nu_x(y)\\
    &=  \sum_{y \in S} \nut_x(y) + \sum_{y \in S} (\nu_x(y) - \nut_x(y))\\
    &= \sum_{y \in S} e^{\eps}\rho(y) + \sum_{y} (\nu_x(y) - \nut_x(y))\\
    &= e^{\eps} \rho(S) + (1 -\sum_{y} \nut_x(y)).
\end{align*}
Then deletion $(\eps, \delta)$-DP of $\R$ implies that $\sum_y \nut_x(y) \geq 1-\delta$.
Observe that this implies that for true randomness we have:
\begin{align*}
    \E_{r' \sim \zo^t} \lb\pi_x(\R_\emptyset(r')) \rb &=  \E_{r' \sim \zo^t} \lb \frac{\nut_x(\R_\emptyset(r'))}{\pr_{r\sim \zo^t}[\R_\emptyset(r) = \R_\emptyset(r')]}\rb\\
    &=  \E_{r' \sim \zo^t} \lb \frac{\nut_x(\R_\emptyset(r'))}{\rho(\R_\emptyset(r'))}\rb\\
    &= \E_{y \sim \rho} \lb \frac{\nut_x(y)}{\rho(y)}\rb\\
    &= \sum_{y \in Y} \rho(y) \cdot \frac{\nut_x(y)}{\rho(y)}\\
    &= \sum_{y\in Y} \nut_x(y) \;\;\in\;\; [1-\delta, 1].
\end{align*}
Using the fact that $\pi_x(y) \in [0,e^\eps]$ we have that
\[ \E_{s' \sim \zo^\ell}[\pi_x(\R_\emptyset(G(s')))] = \int_{0}^{e^\eps} \pr_{s' \sim \zo^\ell}[\pi_x(\R_\emptyset(G(s'))) \geq \theta] d\theta \] and, similarly,
\[\E_{r' \sim \zo^t} \lb\pi_x(\R_\emptyset(r')) \rb = \int_{0}^{e^\eps} \pr_{r' \sim \zo^t} \lb\pi_x(\R_\emptyset(r')) \geq \theta \rb d\theta .\]
Thus, it follows that
\begin{align*}
\left|\E_{s' \sim \zo^\ell}[\pi_x(\R_\emptyset(G(s')))] \right.&\left.- \E_{r' \sim \zo^t} [\pi_x(\R_\emptyset(r')) ] \right| \\
&= \left|\int_{0}^{e^\eps} \left(\pr_{s' \sim \zo^\ell}[\pi_x(\R_\emptyset(G(s'))) \geq \theta] - \pr_{r' \sim \zo^t} \lb\pi_x(\R_\emptyset(r')) \geq \theta \rb\right) d\theta \right|\\
&\leq \int_{0}^{e^\eps} \left|\pr_{s' \sim \zo^\ell}[\pi_x(\R_\emptyset(G(s'))) \geq \theta] - \pr_{r' \sim \zo^t} \lb\pi_x(\R_\emptyset(r')) \geq \theta \rb\right| d\theta \\
&\leq e^\eps \beta,
\end{align*}
where in the last step, we have used the property of the pseudorandom generator that it fools $\D$, and the fact that for $\theta \in [0, e^\eps)$, $\frac{\nut_x(y)}{\pr[\R(\emptyset) = y]} < \theta$ if and only if $\frac{\nu_x(y)}{\pr[\R(\emptyset) = y]} < \theta$. The first part of the claim follows.

For the second part of the claim we first note that deletion $(\eps, \delta)$-DP of $\R$ implies that 
\begin{align*}
    \E_{r' \sim \zo^t} \lb  |1- e^\eps \pi_x(\R_\emptyset(r'))|_+  \rb &= \E_{y \sim \rho} \lb |1- e^\eps \pi_x(y)|_+ \rb \\
    &= \E_{y \sim \rho} \lb |1- e^\eps \pi_x(y)|_+ \rb \\
    &= \sum_{y \in Y} \rho(y) |1- e^\eps \pi_x(y)|_+ \\
    &= \sum_{y \in Y} |\rho(y) - e^\eps \nut_x(y)|_+ \\
    &= \sum_{y \in Y} |\rho(y) - e^\eps \nu_x(y)|_+ \leq \delta .
\end{align*}
Also note that 
\[  \E_{r' \sim \zo^t} \lb  |1- e^\eps \pi_x(\R_\emptyset(r'))|_+  \rb = \int_{0}^1 \pr_{r' \sim \zo^t} \lb  1- e^\eps \pi_x(\R_\emptyset(r')) \geq \theta \rb d\theta = 1 - \int_{0}^1 \pr_{r' \sim \zo^t} \lb  \pi_x(\R_\emptyset(r')) \geq \frac{\theta}{e^\eps} \rb d\theta .\]
Similarly,
\[  \E_{s' \sim \zo^\ell} \lb  |1- e^\eps \pi_x(\R_\emptyset(G(s')))|_+  \rb = 1 - \int_{0}^1 \pr_{s' \sim \zo^\ell} \lb  \pi_x(\R_\emptyset(G(s'))) \geq \frac{\theta}{e^\eps} \rb d\theta .\]

Thus by the same argument as before, the fact that $G$, $\beta$-fools $\D$ implies that 

\[\E_{s' \sim \zo^\ell} \lb \left| 1- e^\eps \pi_x(\R_\emptyset(G(s'))) \right|_+ \rb \leq   \E_{r' \sim \zo^t} \lb  |1- e^\eps \pi_x(\R_\emptyset(r'))|_+  \rb  + \beta \leq \delta + \beta .\]
\end{proof}

We can now give an analogue of Lemma~\ref{lem:prg-privacy} for deletion $(\eps,\delta)$-DP randomizers.
\begin{lem}
\label{lem:prg-privacy-ed}
For  a $t$-samplable deletion $(\eps,\delta)$-DP local randomizer $\R\colon X\to Y$ and $G \colon \zo^\ell \to \zo^t$, let $\D$ denote the following family of tests which take $r'\in \zo^t$ as an input:
 \[ \D \dfn \left\{\left. \ind{\frac{\pr[\R(x) = \R_\emptyset(r')]}{\pr_{r\sim \zo^t}[\R_\emptyset(r) = \R_\emptyset(r')]} \geq \theta}\ \right|\ x\in X,\theta \in [0,e^\eps]\right\} . \]
 If $G$ $\beta$-fools $\D$ where $\delta + e^\eps \beta < 1/2$ then $R[G]$ is a deletion $(\eps+ 2\delta+ 2 e^\eps \beta, \delta +\beta)$-DP local randomizer. Furthermore, for every $\gamma > 0$, $\R[G,\gamma]$ is a deletion $(\eps+ 2\delta+ 2 e^\eps \beta, \delta +\beta)$-DP local randomizer.
\end{lem}
\begin{proof}
  As before, we let the reference distribution $\rho_G$ be the uniform distribution over $\zo^\ell$. Using the definitions in the proof of Lemma~\ref{lem:prg-ed-deletion} we observe that $\R[G](x)$ outputs $s$ with probability:
  
  \[ \mu_x(s) \dfn \frac{\pi(\R_\emptyset(G(s)))}{\sum_{s' \in \zo^\ell}\pi(\R_\emptyset(G(s')))} = \frac{ \frac{\nut_x(\R_\emptyset(G(s)))}{\rho(\R_\emptyset(G(s)))}}{2^{\ell} \cdot \E_{s' \sim \zo^\ell}[\pi_x(\R_\emptyset(G(s')))]} 
   =  \rho_G(s) \cdot \frac{ \frac{\nut_x(\R_\emptyset(G(s)))}{\rho(\R_\emptyset(G(s)))}}{\E_{s' \sim \zo^\ell}[\pi_x(\R_\emptyset(G(s')))]  } .\]
By the definition of $\nut_x$ we have that the numerator satisfies $\frac{\nut_x(\R_\emptyset(G(s)))}{\rho(\R_\emptyset(G(s)))} \leq e^\eps$. In addition, by Lemma~\ref{lem:prg-ed-deletion} the denominator $\E_{s' \sim \zo^\ell}[\pi_x(\R_\emptyset(G(s')))] \geq 1-\delta -e^\eps \beta$. Therefore \[\mu_x(s) \leq \rho_G(s) \cdot \frac{e^\eps}{1-\delta -e^\eps \beta} \leq e^{\eps + 2\delta + e^\eps \beta} \rho_G(s) .\]
For the other side of $(\eps,\delta)$-closeness we simply observe that by the Lemma~\ref{lem:prg-ed-deletion},
\alequn{
\sum_{s \in \zo^\ell} \left|\rho_G(s) - e^{\eps + e^\eps \beta} \mu_x(s)\right|_+ & = \sum_{s \in \zo^\ell} \left|\rho_G(s) - e^{\eps + e^\eps \beta} \rho_G(s) \cdot \frac{\pi_x(\R_\emptyset(G(s)))}{\E_{s' \sim \zo^\ell}[\pi_x(\R_\emptyset(G(s')))]} \right|_+ \\
& \leq \sum_{s \in \zo^\ell} \left|\rho_G(s) - e^{\eps} \rho_G(s) \cdot \pi_x(\R_\emptyset(G(s))) \right|_+ \\
& = \E_{s \sim \zo^\ell}\lb \left|1 - e^{\eps} \pi_x(\R_\emptyset(G(s))) \right|_+ \rb \leq \delta + \beta .
}
Finally to establish that $R[G,\gamma]$ is $(\eps+ 2\delta+ 2 e^\eps \beta, \delta +\beta)$ we, as before, appeal to quasi-convexity.
\end{proof}

To establish the utility guarantees for $\R[G,\gamma]$ we follow the same approach by establishing the utility guarantees for $\R[\ID_t,\gamma]$ and then using the properties of $G$. 
\begin{lem}
\label{lem:rejection-randomizer-ed}
Let $\R$ be a deletion $\eps$-DP $t$-samplable local randomizer. Then for the identity function $\ID_t \colon \zo^t \to \zo^t$ and any $\gamma > 0$ we have that $\R[\ID_t,\gamma]$ is a deletion $\eps$-DP local randomizer and for every $x\in \X$, $\TV(\R_\emptyset(\R[\ID_t,\gamma](x)),\R(x)) \leq \delta + \gamma$.
\end{lem}
\begin{proof}
Conditioned on accepting a sample, $\R[\ID_t,\gamma]$ outputs a sample from the truncated version of the distribution of $\R(x)$. Specifically, $y$ is output with probability $\bar\nu_x(y) \dfn \frac{\nut(y)}{\sum_{y\in Y}\nut(y)}$, where $\nu_x(y) \dfn \pr[\R(x)=y]$ and $\nut_x(y) \dfn \min(\nu_x(y), e^{\eps}\pr[\R(\emptyset) = y])$.
From the proof of Lemma~\ref{lem:prg-ed-deletion}, we know that $\sum_{y \in Y} \nut_x(y) \geq 1-\delta$. Thus
\alequn{ \TV(\nu_x,\bar\nu_x) &= \fr{2} \sum_{y\in Y} | \nu_x(y) - \bar\nu_x(y) | \\
& \leq \fr{2} \sum_{y\in Y} (|\nu_x(y) - \nut_x(y)| + |\nut_x(y) - \bar\nu_x(y)|) \\
& = \fr{2} \sum_{y\in Y} (\nu_x(y) - \nut_x(y) + \bar\nu_x(y) - \nut_x(y)) \leq \delta .}

Truncation of the distribution also reduces the probability that a sample is accepted. Specifically,
\[
  \E_{y \sim \rho} \lb \frac{\nut_x(y)}{e^\eps \rho(y)} \rb =  \sum_{y\in Y} \frac{\nut_x(y)}{e^\eps} \geq \frac{1-\delta}{e^\eps}.
\]
$\R[G,\gamma]$ tries at least $e^\eps \ln(1/\gamma)/(1-\delta)$ samples and therefore, as in the proof of Lemma~\ref{lem:rejection-randomizer}, failure to accept any samples adds at most $\gamma$ to the total variation distance.
\end{proof}

From here we can directly obtain the analogues of Lemma~\ref{lem:prg-compute} and Theorem\ref{thm:main}.

Finally, to deal with the replacement version of $(\eps,\delta)$-DP we combine the ideas we used in Lemmas~\ref{lem:prg-privacy-replace} and \ref{lem:prg-privacy-ed}. The main distinction is a somewhat stronger test that we need to fool in this case.

\begin{lem}
\label{lem:prg-privacy-ed-replacement}
For  a $t$-samplable replacement $(\eps_r,\delta_r)$-DP and deletion $(\eps,\delta)$-DP local randomizer $\R\colon X\to Y$ and $G \colon \zo^\ell \to \zo^t$, let $\D$ and $\D_r$ denote the following families of tests which take $r'\in \zo^t$ as an input:
 \[ \D \dfn \left\{\left. \ind{\frac{\pr[\R(x) = \R_\emptyset(r')]}{\rho(\R_\emptyset(r'))} \geq \theta}\ \right|\ x\in X,\theta \in [0,e^\eps]\right\} ; \]
  \[ \D_r \dfn \left\{\left. \ind{\frac{\nut_x(\R_\emptyset(r'))}{\rho(\R_\emptyset(r'))}  - e^\eps \frac{\nut_{x'}(\R_\emptyset(r'))}{\rho(\R_\emptyset(r'))}  \geq \theta}\ \right|\ x,x'\in X,\theta \in [0,e^{\eps}]\right\} , \]
  where $\rho$ is the reference distribution of $\R$ and $\nut_x(y) \dfn \min(\pr[\R(x)=y], e^{\eps}\rho(y))$.
  If $G$ $\beta$-fools $\D \cup \D_r$ where $\delta + e^\eps \beta < 1/2$ then $R[G]$ is a replacement $(\eps_r + 2\delta+ 3 e^\eps \beta, 2\delta_r +2e^\eps \beta)$-DP local randomizer. Furthermore, for every $\gamma > 0$, $\R[G,\gamma]$ is a $(\eps_r + 2\delta+ 3 e^\eps \beta, 2\delta_r +2e^\eps \beta)$-DP local randomizer.
\end{lem}
\begin{proof}
  First we observe that $\R$ being $(\eps_r,\delta_r)$ replacement DP implies that $\nut_x$ and $\nut_{x'}$ are $(\eps_r,\delta_r)$ close in the following sense:
  \alequn{
  \E_{r'\sim\zo^t} \lb \left| \frac{\nut_x(\R_\emptyset(r'))}{\rho(\R_\emptyset(r'))}  - e^{\eps_r} \frac{\nut_{x'}(\R_\emptyset(r'))}{\rho(\R_\emptyset(r'))} \right|_+ \rb &= \E_{y\sim \rho} \lb \left| \frac{\nut_x(y)}{\rho(y)}  - e^{\eps_r} \frac{\nut_{x'}(y)}{\rho(y)} \right|_+ \rb \\ 
  &=  \sum_{y\in Y}\left| \nut_x(y)  - e^{\eps_r} \nut_{x'}(y)\right|_+\\ 
  & \leq \sum_{y\in Y}\left| \nu_x(y)  - e^{\eps_r} \nu_{x'}(y)\right|_+ \leq \delta_r,
  }
where we used the fact that if $\nu_{x'}(y) > \nut_{x'}(y)$ then $\nut_{x'}(y) = e^\eps \rho(y) \geq \nut_{x}(y)$ and so 
\[\left| \nut_x(y)  - e^{\eps_r} \nut_{x'}(y)\right|_+ = \left| \nut_x(y)  - e^{\eps_r} \nu_{x'}(y)\right|_+ .\]
Using the decomposition
  \[\E_{r'\sim\zo^t} \lb \left| \frac{\nut_x(\R_\emptyset(r'))}{\rho(\R_\emptyset(r'))}  - e^{\eps_r} \frac{\nut_{x'}(\R_\emptyset(r'))}{\rho(\R_\emptyset(r'))} \right|_+ \rb = \int_0^{e^\eps} \pr_{r'\sim\zo^t} \lb  \frac{\nut_x(\R_\emptyset(r'))}{\rho(\R_\emptyset(r'))}  - e^{\eps_r} \frac{\nut_{x'}(\R_\emptyset(r'))}{\rho(\R_\emptyset(r'))} \geq \theta \rb d\theta \]
and the fact that $G$ $\beta$ fools $\D_r$ we obtain that 
  \equ{ \E_{s'\sim\zo^\ell} \lb \left| \frac{\nut_x(\R_\emptyset(G(s')))}{\rho(\R_\emptyset(G(s')))}  - e^{\eps_r} \frac{\nut_{x'}(\R_\emptyset(G(s')))}{\rho(\R_\emptyset(G(s')))} \right|_+ \rb  \leq \delta_r + e^\eps \beta .\label{eq:get-delta}}
  By Lemma~\ref{lem:prg-ed-deletion} we have that for $\pi_x(y) \dfn \frac{\nut_x(y)}{\rho(y)}$ it holds that \[ \zeta_x \dfn \E_{s' \sim \zo^\ell}[\pi_x(\R_\emptyset(G(s')))] \in [1-\delta-e^\eps\beta, 1+e^\eps\beta] .\]
Following the notation in Lemma~\ref{lem:prg-privacy-ed}, we know that the distribution of $\R[G](x)$ is 
  \[\mu_x(s)=  \rho_G(s) \cdot \frac{ \frac{\nut_x(\R_\emptyset(G(s)))}{\rho(\R_\emptyset(G(s)))}}{\E_{s' \sim \zo^\ell}[\pi_x(\R_\emptyset(G(s')))]  }  = \frac{ \rho_G(s) \cdot \nut_x(\R_\emptyset(G(s)))}{\zeta_x \cdot \rho(\R_\emptyset(G(s)))}.\]
   Thus setting $\eps' = \eps_r + 2\delta + 3e^\eps\beta$ we obtain:
  \alequn{
  \sum_{s'\in\zo^\ell} | \mu_x(s')  - e^{\eps'} \mu_{x'}(s') |_+ &= \E_{s'\sim \zo^\ell} \lb \left| \frac{\mu_x(s')}{\rho_G(s')} - e^{\eps'} \frac{ \mu_{x'}(s')}{\rho_G(s')} \right|_+ \rb \\
  & = \E_{s'\sim\zo^\ell} \lb \left| \frac{\nut_x(\R_\emptyset(G(s')))}{\zeta_x \cdot \rho(\R_\emptyset(G(s')))}  - e^{\eps'} \frac{\nut_{x'}(\R_\emptyset(G(s')))}{\zeta_{x'} \cdot \rho(\R_\emptyset(G(s')))} \right|_+ \rb \\
  &= \fr{\zeta_x} \cdot \E_{s'\sim\zo^\ell} \lb \left| \frac{\nut_x(\R_\emptyset(G(s')))}{\rho(\R_\emptyset(G(s')))}  - e^{\eps'} \frac{\zeta_x \cdot \nut_{x'}(\R_\emptyset(G(s')))}{\zeta_{x'} \cdot \rho(\R_\emptyset(G(s')))} \right|_+ \rb \\
  & \leq \fr{1-\delta-e^\eps\beta} \cdot \E_{s'\sim\zo^\ell} \lb \left| \frac{\nut_x(\R_\emptyset(G(s')))}{\rho(\R_\emptyset(G(s')))}  - e^{\eps'} \frac{(1-\delta-e^\eps\beta) \nut_{x'}(\R_\emptyset(G(s')))}{(1+e^\eps\beta) \rho(\R_\emptyset(G(s')))} \right|_+ \rb \\
  &\leq \fr{1-\delta-e^\eps\beta} \cdot \E_{s'\sim\zo^\ell} \lb \left| \frac{\nut_x(\R_\emptyset(G(s')))}{\rho(\R_\emptyset(G(s')))}  - e^{\eps_r} \frac{ \nut_{x'}(\R_\emptyset(G(s')))}{ \rho(\R_\emptyset(G(s')))} \right|_+ \rb \\
  & \leq 2(\delta_r + e^\eps \beta),
  }
where we used that $\frac{1+e^\eps\beta}{1-\delta-e^\eps\beta} \leq e^{2\delta + 3e^\eps\beta}$ and $\fr{1-\delta-e^\eps\beta} \leq 2$ whenever $\delta + e^\eps \beta < 1/2$.
\end{proof}
\fi

\section{Frequency Estimation}
\label{sec:rappor}
In this section we apply our approach to the problem of frequency estimation over a discrete domain. In this problem on domain $X = [k]$, the goal is to estimate the frequency of each element $j \in [k]$ in the dataset. Namely, for $S=(x_1,\ldots,x_n) \in X^n$ we let $c(S) \in \{0,\ldots,n\}^k$ be the vector of the counts of each of the elements in $S$: $c(S)_j = |\{ i \cond x_i=j\}|$. In the frequency estimation problem the goal is to design a local randomizer and a decoding/aggregation algorithm that outputs a vector $\tilde c$ that is close to $c(S)$. Commonly studied metrics are (the expected) $\ell_\infty$, $\ell_1$ and $\ell_2$ norms of $\tilde c - c(S)$. In most regimes of interest, $n$ is large enough and all these errors are essentially determined by the variance of the estimate of each count produced by the randomizer and therefore the choice of the metric does not affect the choice of the algorithm.

The randomizer used in the RAPPOR algorithm \citep{erlingsson2014rappor} is defined by two parameters $\alpha_0$ and $\alpha_1$. The algorithm first converts the input $j$ to the indicator vector of $j$ (also referred to as one-hot encoding). 
It then randomizes each bit in this encoding: if the bit is 0 then $1$ is output with probability $\alpha_0$ (and 0 with probability $1-\alpha_0$) and if the bit is 1 then $1$ is output with probability $\alpha_1$.

For deletion privacy the optimal error is achieved by a symmetric setting $\alpha_0 = 1/(e^\eps + 1)$ and $\alpha_1 = e^\eps/(e^\eps + 1)$ \citep{ErlingssonFMRSTT2020}. This makes the algorithm equivalent to applying the standard binary randomized response to each bit. A simple analysis shows that this results in the standard deviation of each count being $\frac{\sqrt{n}e^{\eps/2}}{e^\eps -1}$ \citep{erlingsson2014rappor,WangBLJ:17}.
For replacement privacy the optimal error is achieved by an asymmetric version in which $\alpha_0 =1/(e^\eps+1)$ but $\alpha_1 = 1/2$. The resulting standard deviation for each count is dominated by $\frac{2\sqrt{n}e^{\eps/2}}{e^\eps -1}$ \citep{WangBLJ:17}. (We remark that several works analyze the symmetric RAPPOR algorithm in the replacement privacy. This requires setting $\alpha_0 = (1-\alpha_1) = 1/(e^{\eps/2} + 1)$ resulting in a substantially worse algorithm than the asymmetric version).

Note that the resulting encoding has $\approx n/(e^\eps+1)$ ones. A closely-related Subset Selection algorithm \citep{wang2016mutual,ye2018optimal} maps inputs to bit vectors of length $k$ with exactly $\lceil \approx n/(e^\eps+1) \rceil$ ones (that can be thought of as a subset of $[k]$). An input $j$ is mapped with probability $\approx 1/2$ to a random subset that contains $j$ and with probability $\approx 1/2$ to a random subset that does not. This results in essentially the same marginal distributions over individual bits and variance bounds as asymmetric RAPPOR. \ifconf\else (This algorithm can also be easily adapted to deletion privacy in which case the results will be nearly identical to symmetric RAPPOR).\fi

\subsection{Pairwise-independent RAPPOR}
While we can use our general result to compress communication in RAPPOR, in this section we exploit the specific structure of the randomizer. Specifically, the tests needed for privacy are fooled if the marginals of the PRG are correct. Moreover the accuracy is preserved as long as the bits are randomized in a pairwise independent way. Thus we can simply use a standard derandomization technique for pairwise independent random variables.
Specifically, to obtain a Bernoulli random variable with bias $\alpha_0$ we will use a finite field $X=\gf{p}$ of size $p$ such that $\alpha_0 p$ is an integer (or, in general, sufficiently close to an integer) and $p$ is a prime larger than $k$. This allows us to treat inputs in $[k]$ as non-zero elements of $\gf{p}$. We will associate all elements of the field that are smaller (in the regular order over integers) than $\alpha_0 p$ with $1$ and the rest with $0$.  We denote this indicator function of the event $z < \alpha_0 p$ by $\bool(z)$. Now for a randomly and uniformly chosen element $z \in \gf{p}$, we have that $\bool(z)$ is distributed as a Bernoulli random variable with bias $\alpha_0$.
\ifconf\else
We note that this approach is a special case of a more general approach is which we associate each $j$ with a non-zero element of an inner product space $\gf{q}^d$, where $q$ is a prime power. This more general approach \ifconf(that we describe in SM)  \else (that we describe in Section \ref{sec:gen-pi-rappor}) \fi allows to reduce some of the computation overheads in decoding.
\fi

As mentioned we will, associate each index $j\in [k]$ with the element $j$ in $\gf{p}$.
We can describe an affine function $\phi$ over $\gf{p}$ using its $2$ coefficients: $\phi_0$ and $\phi_1$ and for $z\in \gf{p}$ we define $\phi(z) = \phi_0 +  z \phi_1$, where addition and multiplication are in the field $\gf{p}$. Each such function encodes a vector in $\gf{p}^k$ as $\phi([k]) \dfn \phi(1),\phi(2),\ldots,\phi(k)$. Let $\Phi \dfn \{\phi \cond \phi \in \gf{p}^{2}\}$ be the family of all such functions. For a randomly chosen function from this family the values of the function on two distinct non-zero values are uniformly distributed and pairwise-independent: for any $j_1\neq j_2 \in [k]$  and $a_1,a_2 \in \gf{p}$ we have that
\ifconf
\else
\[\pr_{\phi \sim \Phi} [\phi(j_1) = a_1 \mbox{ and } \phi(j_2) = a_2] = \pr_{\phi \sim \Phi} [\phi(j_1) = a_1 ] \cdot \pr_{\phi \sim \Phi} [\phi(j_2) = a_2] = \frac{1}{p^2}.\]
\fi
In particular, if we use the encoding of $\phi$ as a boolean vector
\[ \bool(\phi[k]) \dfn \bool(\phi(1)),\bool(\phi(2)),\ldots,\bool(\phi(k))\] then we have that for $\phi \sim \Phi$ and any $j_1 \neq j_2 \in [k]$, $\bool(\phi(j_1))$ and $\bool(\phi(j_2))$ are independent Bernoulli random variables with bias $\alpha_0$.

Finally, for every index $j \in [k]$ and bit $b\in \zo$ we denote the set of functions $\phi$ whose encoding has bit $b$ in position $j$ by $\Phi_{j,b}$:
\equ{\Phi_{j,b} \dfn \{\phi \in \Phi \cond \bool(\phi(j)) = b \} .\label{eq:phi_j}}
We can now describe the randomizer, which we refer to as Pairwise-Independent (PI) RAPPOR for general $\alpha_1> \alpha_0$.
\begin{algorithm}[htb]
	\caption{PI-RAPPOR randomizer}\label{alg:pi-rappor}
	\begin{algorithmic}[1]
		\REQUIRE An index $j \in [k]$, $0<\alpha_0 <\alpha_1<1$, prime $p\geq k+1$ s.t.~$\alpha_0 p \in \mathbb{N}$
        \STATE Sample  $b$ from $\Bern(\alpha_1)$
        \STATE Sample randomly $\phi$ from $\Phi_{j,b}$ defined in eq.~\eqref{eq:phi_j}
		\STATE Send $\phi$
\end{algorithmic}
\end{algorithm}

The server side of the frequency estimation with pairwise-independent RAPPOR consists of a decoding step that converts $\phi$ to $\bool(\phi[k])$ and then the same debiasing and aggregation as for the standard RAPPOR. We describe it as a frequency oracle to emphasize that each count can be computed individually.

\begin{algorithm}[htb]
	\caption{Server-side frequency for PI-RAPPOR}\label{alg:agg-rappor}
	\begin{algorithmic}[1]
		\REQUIRE $0<\alpha_0 <\alpha_1<1$, $k$, index $j \in [k]$ and prime $p > k$. Reports $\phi^1,\ldots,\phi^n$ from $n$ users.
        \STATE $\mbox{sum} = 0$
		\FOR{$i \in [n]$}
            \STATE $\mbox{sum}+=\bool(\phi^i(j))$
        \ENDFOR
        \STATE $\tilde c_j = \frac{\mbox{sum} - \alpha_0 n}{\alpha_1-\alpha_0}$
        \STATE Return $\tilde c_j$
\end{algorithmic}
\end{algorithm}

We start by establishing several general properties of PI-RAPPOR. First we establish that the privacy guarantees for PI-RAPPOR are identical to those of RAPPOR.
\begin{lem}
\label{lem:pi-rappor-privacy}
PI-RAPPOR randomizer (Alg.~\ref{alg:pi-rappor}) is deletion $\max\left\{\frac{\alpha_1}{\alpha_0}, \frac{1-\alpha_0}{1-\alpha_1}\right\}$-DP and replacement $\frac{\alpha_1(1-\alpha_0)}{\alpha_0(1-\alpha_1)}$-DP.
\end{lem}
\ifconf\else
\begin{proof}
While it is easy to analyze the privacy guarantees of PI-RAPPOR directly it is instructive to show that these guarantees follow from our general compression technique. Specifically, there is a natural way to sample from the reference distribution of RAPPOR relative to which our pairwise PRG fools the density tests given in Lemma \ref{lem:prg-privacy}.

To sample from the reference distribution of RAPPOR we pick $k$ values $z_1,\ldots,z_k$ randomly independently and uniformly from $\gf{p}$ and then output $\bool(z_1),\bool(z_2),\ldots,\bool(z_k)$ (we note that samplability is defined using uniform distribution over binary strings  length $t$ as an input but any other distribution can be used instead). By our choice of parameter $p$ and definition of $\bool$, this gives $k$ i.i.d.~samples from $\Bern(\alpha_0)$, which is the reference distribution for RAPPOR. Let $\R$ denote the RAPPOR randomizer. For any $j \in [k]$ and $z' \in \gf{p}^k$ the ratio of densities at $z'$ satisfies:
\[\frac{\pr[\R(j) = \R_\emptyset(z')]}{\pr_{z\sim \gf{p}^k}[\R_\emptyset(z) = \R_\emptyset(z')]}= \begin{cases}
                                                                                                    \frac{\alpha_1}{\alpha_0}, & \mbox{if } \bool(z'_j) = 1 \\
                                                                                                    \frac{1-\alpha_1}{1-\alpha_0}, & \mbox{otherwise}.
                                                                                                  \end{cases} .\]
With probability $\alpha_1$, PI-RAPPOR algorithm  samples $\phi$ uniformly from $\Phi_{j,1}$ and with probability $1-\alpha_1$ PI-RAPPOR algorithm  samples $\phi$ uniformly from $\Phi_{j,0}$. This means that PI-RAPPOR is exactly equal to $\R[G]$, where $G\colon \gf{p}^2 \to \gf{p}^k$ is defined as $G(\phi) = \phi(1),\phi(2),\ldots,\phi(k)$.

Now to prove that PI-RAPPOR has the same deletion privacy guarantees as RAPPOR it suffices to prove that $G$ $0$-fools the tests based on the ratio of densities above. This follows immediately from the fact that $\bool(\phi(j))$ for $\phi \sim \Phi$ is distributed in the same way as $\bool(z_j)$ for $z\sim \gf{p}^k$. 

To prove that PI-RAPPOR has the same replacement privacy guarantees as RAPPOR we simply use the same reference distribution and apply Lemma~\ref{lem:prg-privacy-replace}.
\end{proof}
\fi

Second we establish that the utility guarantees of PI-RAPPOR are identical to those of RAPPOR. This follows directly from the fact that the utility is determined by the variance of the estimate of each individual count in each user's contribution. The variance of the estimate of $c(S)_j$ is a sum of $c(S)_j$ variances for randomization of $1$ and $n-c(S)_j$ variances of randomization of $0$. These variances are identical for RAPPOR and PI-RAPPOR leading to identical \emph{exact} bounds. \ifconf\else Standard results on concentration of sums of independent random variables imply that the bounds on the variance can be translated easily into high probability bounds and also into bounds on the expectation of $\ell_\infty$, $\ell_1$, $\ell_2$ errors.
\fi
\ifconf
\else
\begin{lem}
\label{lem:pi-rappor-utility}
For any dataset $S\in [k]^n$, the estimate $\tilde c$ computed by PI-RAPPOR algorithm (Algs.~\ref{alg:pi-rappor},\ref{alg:agg-rappor}) satisfies:
\begin{itemize}
\item $\E[\tilde c] = c(S)$
\item For all $j\in [k]$,
\[\Var[\tilde c_j] =  c(S)_j \frac{1-\alpha_0 - \alpha_1}{\alpha_1-\alpha_0} + n \frac{\alpha_0(1-\alpha_0)}{(\alpha_1-\alpha_0)^2} \]
For the symmetric case $\alpha_0=1-\alpha_1$ this simplifies to $\Var[\tilde c_j] = n \frac{\alpha_0(1-\alpha_0)}{(1-2\alpha_0)^2}$.
\end{itemize}
In particular, the expected $\ell_2$ squared error is
\[\E \lb \|\tilde c - c(S)\|_2^2\rb = n \frac{1-\alpha_0 - \alpha_1}{\alpha_1-\alpha_0} + n k \frac{\alpha_0(1-\alpha_0)}{(\alpha_1-\alpha_0)^2} .\]
\end{lem}
\begin{proof}
We first note that
\[\tilde c_j = \sum_{i\in [n]} \frac{\bool(\phi^i(j)) - \alpha_0}{\alpha_1 -\alpha_0}, \] where $\phi^i$ is the output of the PI-RAPPOR randomizer on input $x_i$. Thus to prove the claim about the expectation it is sufficient to prove that for every $i$,
\[ \E\lb \frac{\bool(\phi^i(j)) - \alpha_0}{\alpha_1 -\alpha_0} \rb = \ind{x_i = j} \]
and to prove the claim for variance it is sufficient to prove that
\[ \Var\lb \frac{\bool(\phi^i(j)) - \alpha_0}{\alpha_1 -\alpha_0} \rb = \ind{x_i = j} \frac{1-\alpha_0 - \alpha_1}{\alpha_1-\alpha_0} + \frac{\alpha_0(1-\alpha_0)}{(\alpha_1-\alpha_0)^2} .\]

If $x_i = j$ then both of these claims follow directly from the fact that, by definition of PI-RAPPOR randomizer, in this case the distribution of $\bool(\phi^i(j))$ is $\Bern(\alpha_1)$.

If, on the other hand $x_i \neq j$, we use pairwise independence of $\bool(\phi(x_i))$ and $\bool(\phi(j))$ for $\phi \sim \Phi$ to infer that conditioning the distribution $\bool(\phi(x_i)) = b$ (for any $b$) does not affect the distribution of $\bool(\phi(x_i))$. Thus, if $x_i \neq j$ then $\bool(\phi^i(j))$ is distributed as $\Bern(\alpha_0)$ and we can verify the desired property directly.

Finally,
\[\E \lb \|\tilde c - c(S)\|_2^2\rb = \E \lb  \sum_{j\in[k]} (\tilde c_j - c(S)_j)^2 \rb = \sum_{j\in[k]} \Var[\tilde c_j] =  n \frac{1-\alpha_0 - \alpha_1}{\alpha_1-\alpha_0} + n k \frac{\alpha_0(1-\alpha_0)}{(\alpha_1-\alpha_0)^2} \]
\end{proof}
\fi

\ifconf
\else
For RAPPOR these bounds are stated in \citep{WangBLJ:17} who also demonstrate that optimizing $\alpha_0$ and $\alpha_1$ subject to the replacement privacy parameter being $\eps$ while ignoring the first term in the variance (since it is typically dominated by the second term) leads to the asymmetric version  $\alpha_0 =1/(e^\eps+1)$ and $\alpha_1 = 1/2$. For deletion privacy the optimal setting of  $\alpha_0 = 1-\alpha_1 =1/(e^\eps+1)$ follows from standard optimization of the binary randomized response.
Thus we obtain the following utility bounds for $\eps$-DP versions of PI-RAPPOR.
\fi
\begin{cor}
\label{cor:pi-rappor-utility-dp-deletion}
For any $\eps > 0$ and a setting of $p$ that ensures that $p/(e^\eps+1) \in \mathbb{N}$ we have that PI-RAPPOR for $\alpha_0 = 1-\alpha_1 =1/(e^\eps+1)$ satisfies deletion $\eps$-DP and for every dataset $S\in [k]^n$, the estimate $\tilde c$ computed by PI-RAPPOR satisfies:
$\E[\tilde c] = c(S)$, for all $j\in [k]$, $\Var[\tilde c_j] = n \frac{e^\eps}{(e^\eps -1)^2}$ and $\E \lb \|\tilde c - c(S)\|_2^2\rb = n k \frac{e^\eps}{(e^\eps -1)^2}$.
\end{cor}

\begin{cor}
\label{cor:pi-rappor-utility-dp-replace}
For any $\eps > 0$ and a setting of $p$ that ensures that $p/(e^\eps+1) \in \mathbb{N}$ we have that PI-RAPPOR for $\alpha_0 =1/(e^\eps+1)$ and $\alpha_1 = 1/2$ is replacement $\eps$-DP and for every dataset $S\in [k]^n$, the estimate $\tilde c$ computed by PI-RAPPOR satisfies:
$\E[\tilde c] = c(S)$, for all $j\in [k]$, $\Var[\tilde c_j] = c(S)_j + n \frac{4 e^\eps}{(e^\eps -1)^2}$ and $\E \lb \|\tilde c - c(S)\|_2^2\rb = n + nk \frac{4 e^\eps}{(e^\eps -1)^2}$.
\end{cor}
\ifconf\else
Note that in the setting where $S$ is sampled i.i.d.~from some distribution over $[k]$ defined by frequencies $f_1,\ldots,f_k$, the term $c(S)_j$ in the variance is comparable to sampling variance. This is true since $c(S)_j \approx n f_j$ and for a sum of $n$ Bernoulli random variable with bias $f_j \ll 1$, the variance is $n f_j(1-f_j) \approx nf_j$. In most practical regimes of frequency estimation with LDP, sampling error is much lower than error introduced by the local randomizers. This justifies optimization of parameters based on the second term alone.\fi

\ifconf
\else
Finally, we analyze the computational and communication costs of PI-RAPPOR. We first bound these for the client.
\begin{lem}
\label{lem:rappor-communication-computation}
PI-RAPPOR randomizer (Alg.~\ref{alg:pi-rappor}) can be implemented in $\tilde O(\log p)$ time and uses $2 \lceil \log_2 p\rceil $ bits of communication.
\end{lem}
\ifconf\else \begin{proof}
  Any function $\phi \in \Phi$ is represented by two elements from $\gf{p}$ which implies the claimed bound on the communication cost. The running time of PI-RAPPOR is dominated by the time to pick a random and uniform element in $\Phi_{j,b}$. This can be done by picking $\phi_1 \in \gf{p}$ randomly and uniformly. We then need to pick $\phi_0$ randomly and uniformly from the set
  $\{\phi_0 \cond \bool(\phi(j)) = b \}$. Given the result of multiplication $j \phi_1$ this can be done in $O(\log p)$ time. For example for $b=1$ this set is equal to $\{ - j \phi_1, - j \phi_1 +1, \ldots, - j \phi_1 + \alpha_0 p - 1\}$ where all arithmetic operations are in $\gf{p}$. The set consists of at most two contiguous ranges of integers and thus a random and uniform element can be chosen in $O(\log p)$ time.   Multiplication in $\gf{p}$ can be done in $O(\log(p)\cdot (\log\log p)^2)$ (\eg \cite{menezes2018handbook}) but in most practical settings standard Montgomery modular multiplication that takes $O(\log^2(p))$ time would be sufficiently fast.
\end{proof}\fi

The analysis of the running time of decoding and aggregation is similarly straightforward since decoding every bit of message takes time that is dominated by the time of a single multiplication in $\gf{p}$.
\begin{lem}
\label{lem:rappor-decode-computation}
For every $j\in k$, the server-side of PI-RAPPOR (Alg.~\ref{alg:agg-rappor}) computes $\tilde c_j$ in time $\tilde O(n\log p)$. In particular, the entire histogram is computed in time $\tilde O(kn\log p)$.
\end{lem}

Note that the construction of the entire histogram on the server is relatively expensive. In Section~\ref{sec:gen-pi-rappor} we show an alternative algorithm that runs faster when $k \gg n$. For comparison we note that aggregation in the compression schemes in \citep{Acharya:2019} and \citep{chen2020breaking} can be done in $\tilde O(n+k)$. However these schemes require $\Omega(k)$ computation on each client and thus the entire system also performs $\Omega(nk)$ computation. They also do not give a frequency oracle since the decoding time of even a single message is linear in $k$.

Finally we need to discuss how to pick $p$. In addition to the condition that is $p$ a prime larger than $k$, our algorithm requires that $\alpha_0 p$ be an integer. We observe that while, in general, we cannot always guarantee that $\alpha_0 = p/(e^\eps +1)$, by picking $p$ that is a sufficiently large multiple of $\max\{e^\eps,1/\eps\}$ we get an $\eps'$-DP PI-RAPPOR algorithm for $\eps'$ that is slightly smaller than $\eps$ (which also implies that its utility is slightly worse). We make this formal below.
\begin{lem}
\label{lem:pi-rappor-approx}
There exists a constant $c_0$ such that for any $\eps > 0$, $k\in \mathbb{N}$, $\Delta > 0$ and any prime $p \geq c_0 \max\{e^\eps, 1/\eps\}/\Delta$ we have that symmetric PI-RAPPOR with parameter $\alpha_0 = \lceil p/(e^\eps + 1) \rceil/p$ satisfies deletion $\eps$-DP and outputs an estimate that satisfies:
for all $j\in [k]$, $\Var[\tilde c_j] \leq n \frac{(1+\Delta) e^\eps }{(e^\eps -1)^2}$. Further, PI-RAPPOR with $\alpha_0 = \lceil p/(e^\eps + 1) \rceil/p$ and $\alpha_1 =1/2$ satisfies replacement $\eps$-DP and outputs an estimate that satisfies:
for all $j\in [k]$, $\Var[\tilde c_j] = c(S)_j + n \frac{4 (1+\Delta) e^\eps}{(e^\eps -1)^2}$.
\end{lem}
\begin{proof}
We first note that by our definition, $\alpha_0 p = \lceil p/(e^\eps + 1) \rceil$ and therefore is an integer (as required by PI-RAPPOR).
We denote by $\eps' = \ln(1-p/\alpha_0)$  (so that $\alpha_0 = 1/(e^{\eps'} + 1)$ and note that
$\eps' \leq \eps$. Thus the symmetric PI-RAPPOR satisfies $\eps$-DP. We now note that $|1/(e^{\eps'} + 1) - 1/(e^{\eps} + 1)| \leq 1/p$. This implies that the bound on variance of PI-RAPPOR satisfies:
\ifconf
\begin{align*}
    \Var[\tilde c_j] &= n \frac{\alpha_0(1-\alpha_0)}{(1-2\alpha_0)^2}  = n \frac{\fr{e^{\eps'}+1}(1-\fr{e^{\eps'}+1})}{(1-2\fr{e^{\eps'}+1} )^2} \\
    &\leq  n \frac{ ( \fr{e^{\eps}+1} + \fr{p}) (1-\fr{e^{\eps}+1})}{(1-2\fr{e^{\eps'}+1} -\frac{2}{p} )^2} .
\end{align*}
\else
\[\Var[\tilde c_j] = n \frac{\alpha_0(1-\alpha_0)}{(1-2\alpha_0)^2}  = n \frac{\fr{e^{\eps'}+1}(1-\fr{e^{\eps'}+1})}{(1-2\fr{e^{\eps'}+1} )^2} \leq  n \frac{ ( \fr{e^{\eps}+1} + \fr{p}) (1-\fr{e^{\eps}+1})}{(1-2\fr{e^{\eps'}+1} -\frac{2}{p} )^2} . \]
\fi
If $\eps \leq 1$ then $\fr{e^{\eps'}+1}\geq \fr{e+1}$ and $1-2\fr{e^{\eps'}+1} \geq \frac{\eps}{e+1}$. Thus the addition/subtraction of $1/p$ to these quantities for $p \geq c_0/(\eps\Delta)$ increases the bound by at most a multiplicative factor $(1+\Delta)$ (for a sufficiently large constant $c_0$).

Otherwise (if $\eps > 1$), then $\fr{e^{\eps'}+1}\geq \fr{e^\eps}$ and $1-2\fr{e^{\eps'}+1} \geq \frac{e-1}{e+1}$. Thus the addition/subtraction of $1/p$ to these quantities for $p \geq c_0 e^\eps/\Delta$ increases the bound by at most a multiplicative factor $(1+\Delta)$ (for a sufficiently large constant $c_0$).

The analysis for replacement DP is analogous.
\end{proof}
In practice, setting $\Delta = 1/100$ will make the loss of accuracy insignificant. Thus we can conclude that PI-RAPPOR with $p \geq c_1 \max\{k, e^\eps,1/\eps\}$ for a sufficiently large constant $c_1$ achieves essentially the same guarantees as RAPPOR. This means that the communication cost of PI-RAPPOR is $2\log_2(\max\{k, e^\eps,1/\eps\}) + O(1)$. Also we are typically interested in compression when $k \gg \max\{e^\eps, 1/\eps\}$ and in such case the communication cost is $2\log_2(k) + O(1)$.
\fi

\section{Mean Estimation}
\label{sec:mean}
In this section, we consider the problem of mean estimation in $\ell_2$ norm, for $\ell_2$-norm bounded vectors. Formally, each client has a vector $\vx_i \in \mathbb{B}^d$, where $\mathbb{B}^d \dfn \{\vx \in \R^d \cond \|\vx\|_2 \leq 1 \}$. Our goal is to compute the mean of these vectors privately, and we measure our error in the $\ell_2$ norm. In the literature this problem is often studied in the statistical setting where $\vx_i$'s are sampled i.i.d.~from some distribution supported on $\mathbb{B}^d$ and the goal is to estimate the mean of this distribution. In this setting, the expected squared $\ell_2$ distance between the mean of the distribution and the mean of the samples is at most $1/n$ and is dominated by the privacy error in the regime that we are interested in ($\eps < d$).

In the absence of communication constraints and $\eps < d$, the optimal $\eps$-LDP protocols for this problem achieve an expected squared $\ell_2$ error of $\Theta(\frac{d}{n \min(\eps, \eps^2)})$ \citep{duchi2018minimax,DuchiR19}. Here and in the rest of the section we focus on the replacement DP both for consistency with existing work and since for this problem the dependence on $\eps$ is linear (when $1 <\eps < d$) and thus the difference between replacement and deletion is less important.

If one is willing to relax to $(\eps,\delta)$ or concentrated differential privacy \citep{DworkR16,BunS16,Mironov17} guarantees, then standard Gaussian noise addition achieves the asymptotically optimal bound. When $\eps \leq 1$, the randomizer of~\citet{duchi2018minimax} (which we refer to as $\PrivHemi$) also achieves the optimal $O(\frac{d}{n\eps^2})$ bound.
Recent work of~\citet{ErlingssonFMRSTT2020} gives a low-communication version of $\PrivHemi$. Specifically, in the context of federated optimization they show that $\PrivHemi$ is equivalent to sending a single bit and a randomly and uniformly generated unit vector. This vector can be sent using a seed to a PRG. \citet{bhowmick2019protection} describe the $\PrivUnit$ algorithm that achieves the optimal bound also when $\eps > 1$. Unfortunately, $\PrivUnit$ has high communication cost of $\Omega(d)$.

By applying Theorem~\ref{thm:main} to $\PrivUnit$ or Gaussian noise addition, we can immediately obtain a low communication algorithm with negligible effect on privacy and utility. This gives us an algorithm that communicates a single seed, and has the asymptotically optimal privacy utility trade-off. Implementing $\PrivUnit$ requires sampling uniformly from a spherical cap $\{\vv \cond \|\vv\|_2 =1, \la \tilde \vx, \vv \ra \geq \alpha \}$ for $\alpha\approx \sqrt{\eps/d}$. Using standard techniques this can be done with high accuracy using $\tilde O (d)$ random bits and $\tilde O (d)$ time. Further, for every $\vx$ the resulting densities can be computed easily given the surface area of the cap. Overall rejection sampling can be computed in $\tilde O (d)$ time. Thus this approach to compression requires time $\tilde O(e^\eps d)$. This implies that given an exponentially strong PRG $G$, we can compress $\PrivUnit$ to $O(\log (dn) + \eps)$ bits with negligible effects on utility and privacy.  In most settings of interest, the computational cost $\tilde O(e^\eps d)$ is not much larger than the typical cost of computing the vector itself, e.g. by back propagation in the case of gradients of neural networks (e.g.~$\eps=8$ requires $\approx 3000$ trials in expectation).

We can further reduce this computational overhead. We show a simple reduction from the general case of $\eps > 1$ to a protocol for $\eps' = \eps/m$ that preserves asymptotic optimality, where $m \leq 2 \eps$ is an integer. The algorithm simply runs $m$ copies of the $\eps'$-DP randomizer and sends all the reports. The estimates produced from these reports are averaged by the server. This reduces the expected number of rejection sampling trials to $m e^{\eps/m}$.
\ifconf\else
Below we describe the reduction and state the resulting guarantees.

\begin{lem}
\label{lem:repetition}
Assume that for some $\eps > 0$ there exists a local $\eps$-DP randomizer $\R_\eps\colon \mathbb{B}^d \to Y$ and a decoding procedure $\decode \colon Y \to \R^d$ that for all $\vx \in \mathbb{B}^d$, satisfies:
$\E[\decode(\R_\eps(\vx))] = \vx$ and $\E[\|\decode(\R_\eps(\vx)) - \vx\|_2^2] \leq \alpha_\eps$. Further assume that $\R_\eps$ uses $\ell$ bits of communication and runs in time $T$. Then for every integer $m\geq 2$ there is a local $(m\eps)$-DP randomizer $\R_{\eps}^m\colon \mathbb{B}^d \to Y^m$ and decoding procedure $\decode^m\colon Y^m \to \R^d$ that uses $m \ell$ bits of communication, runs in time $m T$ and for every $\vx \in \mathbb{B}^d$ satisfies:
$\E[\decode^m(\R_\eps'(\vx))] = \vx$ and $\E[\|\decode^m(\R_\eps^m(\vx)) - \vx\|_2^2] \leq \frac{\alpha_\eps}{m}$.

In particular, if for every $\eps \in (1/2,1]$, $\alpha_\eps \leq \frac{c d}{\eps^2}$ for some constant $c$, then for every $\eps > 0$ there is a local $\eps$-DP randomizer $\R_\eps'$ and decoding procedure $\decode'$ that uses $\lceil \eps \rceil \ell$ bits of communication, runs in time $\lceil \eps \rceil T$ and for every $\vx \in \mathbb{B}^d$ satisfies:
$\E[\decode'(\R_\eps'(\vx))] = \vx$ and $\E[\|\decode'(\R_\eps'(\vx)) - \vx\|_2^2] \leq \frac{2 c d}{\min\{\eps,\eps^2\}}$.
\end{lem}
\begin{proof}
The randomizer $\R_\eps^m(\vx)$ runs $\R_{\eps}(\vx)$ $m$ times independently to obtain $y_1,\ldots, y_m$ and outputs these values. To decode we define $\decode^m(y_1,\ldots,  y_m) \dfn \fr{m} (\decode(y_1)+\cdots+\decode(y_m))$.
By (simple) composition of differential privacy, $\R^m_\eps$ is $(\eps m)$-DP. The utility claim follows directly from linearity of expectation and independence of the estimates:
\[\E[\|\decode^m(\R_\eps^m(\vx)) - \vx\|_2^2]\\ = \fr{m} \cdot \E[\|\decode(\R_\eps(\vx)) - \vx\|_2^2] \leq \frac{\alpha_\eps}{m} .\]

For the second part of the claim we define $\R'_\eps$ as follows. For $\eps \leq 1$, $\R_\eps'(\vx)$ just outputs $\R_\eps(\vx)$ and in this case $\decode'$ is the same as $\decode$. For $\eps > 1$, we let $m = \lceil \eps \rceil$ and apply the lemma to $\R_{\eps'}$ for $\eps' = \eps / \lceil \eps \rceil$. Note that  $\eps' \in (1/2,1)$ and therefore the resulting bound on variance is
 \[ \E[\|\decode'(\R_\eps'(\vx)) - \vx\|_2^2]  \leq \fr{\lceil \eps \rceil} \frac{c d}{\eps'^2} = \frac{c d}{\eps \eps'} \leq \frac{2 c d}{\eps}.\]
\end{proof}

For example, by using the reduction in Lemma~\ref{lem:repetition}, we can reduce the computational cost to $\tilde O(\lceil \eps \rceil d)$ while increasing the communication to $O(\lceil \eps \rceil \log d)$. The server side reconstruction now requires sampling and averaging $n \lceil \eps \rceil$ $d$-dimensional vectors. Thus the server running time is $\tilde O(n d \eps)$.

\fi

This reduction allows one to achieve different trade-offs between computation, communication, and closeness to the accuracy of the original randomizer. As an additional benefit, we no longer need an LDP randomizer that is optimal in the $\eps> 1$ regime. We can simply use $m = \lceil\eps \rceil$ and get an asymptotically optimal algorithm for $\eps > 1$ from any algorithm that is asymptotically optimal for $\eps' \in [1,1/2]$. In particular, instead of $\PrivUnit$ we can use the low communication version of $\PrivHemi$ from \citep{ErlingssonFMRSTT2020}. This bypasses the need for our compression algorithm and makes the privacy guarantees unconditional.

\begin{rem}
We remark that the compression of $\PrivUnit$ can be easily made unconditional.  The reference distribution $\rho$ of $\PrivUnit$ is uniform over a sphere of some radius $B(d,\eps)=O(\sqrt{d/\min(\eps,\eps^2)})$. It is not hard to see that both the privacy and utility guarantees of $\PrivUnit$ are preserved by any PRG $G$ which preserves $\pr_{\vv \sim \rho}[\langle \vx, \vv \rangle \geq \theta]$ for every vector $\vx$ sufficiently well (up to some $1/\mbox{poly}(d,n,e^\eps/\eps)$ accuracy). Note that these tests are halfspaces and have VC dimension $d$. Therefore by the standard $\epsilon$-net argument, a random sample of size $O(d B(d,\eps)/\gamma^2)$ from the reference distribution will, with high probability, give a set of points $S$ that $\gamma$-fools the test (for any $\gamma > 0$). By choosing $\gamma = 1/\mbox{poly}(d,n,e^\eps/\eps)$ we can ensure that the effect on privacy and accuracy is negligible (relative to the error introduced due to privacy). Thus one can compress the communication to $\log_2(|S|) = O(\log (dn/\eps) + \eps)$ bits unconditionally (with negligible effect on accuracy and privacy).
\end{rem}
\ifconf
\else

\newcommand{\SQKR}{\ensuremath{\texttt{SQKR}}}
\subsection{Empirical Comparison of Mean Estimation Algorithms}
\eat{
There are at least two known algorithms for $\ell_2$ mean estimation that have low communication. \citet{ErlingssonFMRSTT2020} show that the $\PrivHemi$ algorithm~\citep{duchi2018minimax} can be compressed so as to send one bit (and a seed). Further, \citet{chen2020breaking} show an algorithm called $\SQKR$ that communicates $\eps$ bits (and a seed). In comparison, our approach allows us to compress any algorithm, and hence we can losslessly compress the $\PrivUnit$ algorithm~\citep{bhowmick2019protection}, which is asymptotically optimal. In this section, we empirically compare these algorithms.

We consider the performance of these algorithms on vectors of norm exactly $1$. This is essentially without loss of generality for high-dimensional vectors as a vector inside the unit ball in $\Re^d$ can be lifted to a unit vector in $\Re^{d+1}$, and projecting back a report to $d$ dimensions preserves the mean and only reduces the variance.

}

  While $\PrivUnit$, $\SQKR$ and the repeated version of $\PrivHS$ (using Lemma~\ref{lem:repetition}) are asymptotically optimal, the accuracy they achieve in practice may be different. Therefore we empirically compare these algorithms. 
In our first comparison we consider four algorithms. The $\PrivHemi$ algorithm outputs a vector whose norm is fully defined by the parameters $d, \eps$: the output vector has norm $B(d,\eps) = \frac{e^{\eps}+1}{e^\eps-1}\frac{\sqrt{\pi}}{2} \frac{d \Gamma(\frac{d-1}{2} + 1)}{\Gamma(\frac{d}{2}+1)}$. The variance is then easily seen to be $(B^2 + 1 \pm 2B)/n$ when averaging over $n$ samples. For large dimensional settings of interest, $B \gg 1$ so this expression is very well approximated by $B^2/n$ and we use this value as a proxy. As an example, for $d=2000, \eps=8$, $B^2 \approx 3145$ so that the proxy is accurate up to $3.5\%$ (and the proxy is even more accurate when $d$ is significantly larger which is the typical setting where compression).
   For $\SQKR$, we use the implementation provided by the authors at \citep{KashinImplement2021} (specifically, second version of the code that optimizes some of the parameters). We show error bars for the empirical squared error based on 20 trials.

    The $\PrivUnit$ algorithm internally splits its privacy budget $\eps$ into two parts $\eps_0, \eps_1 = 1-\eps_0$. As in the case of $\PrivHemi$, the output of $\PrivUnit$ (for fixed $d, \eps_0, \eps_1$) has a fixed squared norm which is the proxy we use for variance. We first consider the default split used in the experiments in~\citep{bhowmick2019protection} and refer to it as $\PrivUnit$. In addition, we optimize the splitting so as to minimize the variance proxy, by evaluating the expression for the variance proxy as a function of the $\theta = \eps_0 / \eps$, for $101$ values of  $\theta = 0.00, 0.01, 0.02,\ldots, 0.99, 1.0$. We call this algorithm $\PrivUnitOpt$. Note that since we are optimizing $\theta$ to minimize the norm proxy, this optimization is data-independent and need only be done once for a fixed $\eps$. For both variants of $\PrivUnit$, we use the norm proxy in our evaluation; as discussed above, in high-dimensional settings of interest, the proxy is nearly exact.


Figure~\ref{fig:l2error} (Left) compares the expected squared error of these algorithms for $d=1,000$, $n=10,000$ and $\eps$ taking integer values from $1$ to $8$. These plots show both $\PrivUnit$ and $\PrivUnitOpt$ are more accurate than $\PrivHemi$ and $\SQKR$ in the whole range of parameters
While $\PrivHemi$ is competitive for small $\eps$, it does not get better with $\eps$ for large $\eps$. $\SQKR$ consistently has about $5\times$ higher expected squared error than $\PrivUnitOpt$ and about $2.5\times$ higher error compared to $\PrivUnit$. Thus in the large $\eps$ regime, the ability to compress $\PrivUnitOpt$ gives a $5\times$ improvement in error compared to previous compressed algorithms. We also observe that $\PrivUnitOpt$ is noticeably better than $\PrivUnit$. Our technique being completely general, it will apply losslessly to any other better local randomizers that may be discovered in the future.

As discussed earlier, one way to reduce the computational cost of compressed $\PrivUnitOpt$ is to use Lemma~\ref{lem:repetition}. 
For instance, instead of running $\PrivUnitOpt$ with $\eps=8$, we may run it twice with $\eps=4$ and average the results on the server. Asymptotically, this gives the same expected squared error and we empirically evaluate the effect of such splitting on the expected error. Figure~\ref{fig:l2error} (Right) shows the results for $\PrivHemi, \PrivUnit$ and $\PrivUnitOpt$. We plot the single repetition version of $\SQKR$ for comparison. The $\SQKR$ algorithm does not get more efficient for smaller $\eps$ and thus splitting it makes it worse in every aspect. As its error grows quickly with splitting, we do not plot the split version of $\SQKR$ in these plots. The results demonstrate that splitting does have some cost in terms of expected squared error, and going from $\eps=8$ to two runs of $\eps=4$ costs us about $2\times$ in expected squared error, and that the error continues to increase as we split more. These results can inform picking an appropriate point on the computation cost-error tradeoff and suggest that for $\eps$ around $8$, the choice in most cases will be between not splitting and splitting into two mechanisms. Note that even with two or three repetitions, $\PrivUnitOpt$ has $2-3\times$ smaller error compared to $\PrivHemi$ and $\SQKR$. For $\PrivHemi$, the sweet spot seems to be splitting into multiple mechanisms each with $\eps \approx 2$. 


\begin{figure}
\begin{subfigure}{.5\textwidth}
  \centering
  \includegraphics[width=.8\linewidth]{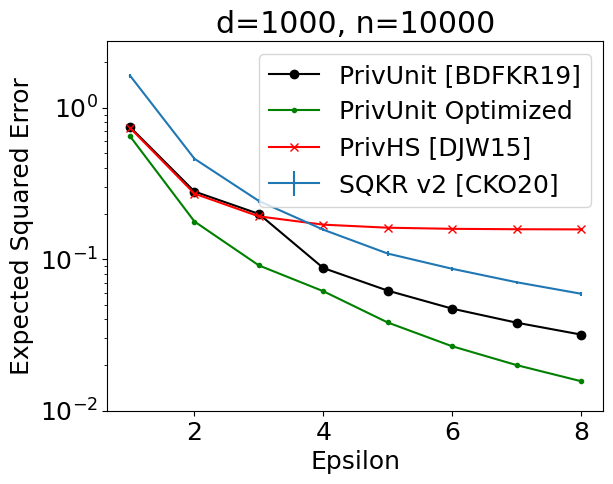}
\end{subfigure}%
\begin{subfigure}{.5\textwidth}
  \centering
  \includegraphics[width=.8\linewidth]{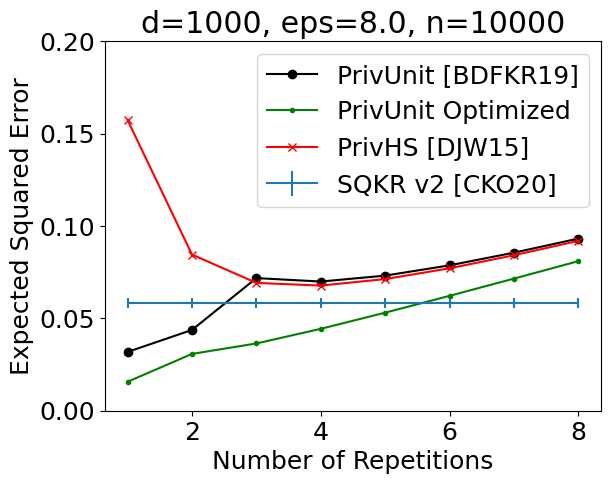}
\end{subfigure}%
    \caption{(Left) Expected $\ell_2^2$ error of mechanisms $\PrivHemi, \PrivUnit, \PrivUnitOpt$ and $\SQKR$ for values of $\eps$ between $1$ and $8$. (Right) Expected $\ell_2^2$ error of mechanisms $\PrivHemi$, $\PrivUnit$ and $\PrivUnitOpt$ for a total $\eps=8$, as a function of the number of repetitions of the mechanism with a proportionately smaller $\eps$. The $\SQKR$ v2 line is for a single run with $\eps=8$ without splitting. Both plots use $n=10,000, d=1,000$, and $95\%$ error bars for $\SQKR$ are computed based on 10 trials.}
    \label{fig:l2error}
\end{figure}


\fi


\printbibliography
\appendix

\section{A generalization of PI-RAPPOR}
\label{sec:gen-pi-rappor}
In the more general version of RAPPOR we let $q$ by any prime power and assume that $\alpha_0 q$ is an integer. We will rely on some arbitrary order on the elements of $\gf{q}$ and denote $\alpha_0 q$ smallest elements as $F_1$. We denote the indicator function of the event $z\in F_1$ by $\bool(z)$. As before, for a randomly and uniformly chosen element $z \in F$ we have that $\bool(z)$ is distributed as Bernoulli random variable with bias $\alpha_0$. For efficiency, the order needs to be chosen in a way that allows computing $\bool(z)$ and generating a random element of the field in some range in $O(\log q)$ time.

We will associate each index $j\in [k]$ with a distinct \emph{non-zero} element of $z(j) \in \gf{q}^d$, where  $d \dfn \lceil \log_q (k+1) \rceil$ (in particular, $q^d \leq k+1 < q^{d+1}$). We can describe an affine function $\phi$ over $\gf{q}^d$ using the vector of its $d+1$ coefficients: $\phi_0,\ldots,\phi_d$ and for $z\in \gf{q}^d$ we define $\phi(z) = \phi_0 + \sum_{u \in [d]} z_u \phi_u$, where addition and multiplication are in the field $\gf{q}$. For brevity we will also overload $\phi(j) \dfn \phi(z(j))$.
Each such function encodes a vector in $\gf{q}^k$ as $\phi([k]) \dfn \phi(1),\phi(2),\ldots,\phi(k)$.
The family of functions defined by all $d+1$ tuples in $\gf{q}$ is defined as $\Phi \dfn \{\phi \cond \phi \in \gf{q}^{d+1}\}$. For a randomly chosen function from this family the values of the function on two distinct non-zero values uniformly distributed and pairwise independent: for any $j_1\neq j_2 \in [k]$  and $a_1,a_2 \in \gf{q}$ we have that
\[\pr_{\phi \sim \Phi} [\phi(j_1) = a_1 \mbox{ and } \phi(j_2) = a_2] = \pr_{\phi \sim \Phi} [\phi(j_1) = a_1 ] \cdot \pr_{\phi \sim \Phi} [\phi(j_2) = a_2] = \frac{1}{q^2}.\]

For every index $j \in [k]$ and bit $b\in \zo$ we denote the set of functions $\phi$ whose encoding has bit $b$ in position $j$ by $\Phi_{j,b}$:
\equ{\Phi_{j,b} \dfn \{\phi \in \Phi \cond \bool(\phi(j)) = b \} .\label{eq:phi_j-gen}}
The generalization of the PI-RAPPOR randomizer is described below.
\begin{algorithm}[htb]
	\caption{General PI-RAPPOR randomizer}\label{alg:pi-rappor-gen}
	\begin{algorithmic}[1]
		\REQUIRE An index $j \in [k]$, $0<\alpha_0 <\alpha_1<1$, prime power $q$ s.t.~$\alpha_0 q \in \mathbb{N}$
        \STATE $d = \lceil \log_q (k+1) \rceil$
        \STATE Sample Bernoulli $b$ with bias $\alpha_1$
        \STATE Sample randomly $\phi$ from $\Phi_{j,b}$ defined in eq.~\eqref{eq:phi_j-gen}
		\STATE Send $\phi$
\end{algorithmic}
\end{algorithm}

The server side of the frequency estimation can be done exactly as before. We can convert $\phi$ to $\bool(\phi(j))$ and then aggregate the results. In addition, we describe an alternative algorithm that runs in time $\tilde O(k \|\Phi\| + n)$. This algorithm is faster than direct computation when $|\Phi| < n$. In this case we can first count the number of times each $\phi$ is used and then decode each $\phi$ only once.
This approach is based on an idea in \citep{bassily2020practical} which also relies on pairwise independence to upper bound the total number of encodings. We note that in the simpler version of PI-RAPPOR $|\Phi| \geq (k+1)^2$, whereas in the generalized version $|\Phi|$ can be as low $(k+1) \cdot (e^\eps + 1)$.
\begin{algorithm}[htb]
	\caption{Private histograms with PI-RAPPOR}\label{alg:agg-rappor-gen}
	\begin{algorithmic}[1]
		\REQUIRE $0<\alpha_0 <\alpha_1<1$, prime power $q$
		\STATE Receive $\phi^1,\ldots,\phi^n$ from $n$ users.
		\FOR{$\phi \in \Phi$}
            \STATE $n_\phi = 0$
        \ENDFOR
		\FOR{$i \in [n]$}
            \STATE $n_{\phi^i} += 1$
        \ENDFOR

        \STATE $\mbox{sum} = -\alpha_0,\ldots,-\alpha_0$
		\FOR{$\phi \in \Phi$}
            \STATE $\mbox{sum} += n_\phi \cdot \bool(\phi^i)$
        \ENDFOR
        \STATE $\tilde c = \frac{1}{\alpha_1-\alpha_0} \mbox{sum}$
        \STATE Return $\tilde c$
\end{algorithmic}
\end{algorithm}

It is easy to see that pairwise-independence implies that privacy and utility guarantees of the generalized PI-RAPPOR are the same as for the simpler version we described before. The primary difference is in the computational costs.
\begin{lem}
\label{lem:rappor-communication-computation-gen}
PI-RAPPOR randomizer (Alg.~\ref{alg:pi-rappor-gen}) can be implemented in $\tilde O(\log k)$ time and uses $\lceil \log_2 |\Phi|\rceil \leq \log_2 k + 2\log_2 q + 1$ bits of communication.
\end{lem}
\begin{proof}
  The running time of PI-RAPPOR is dominated by the time to pick a random and uniform element in $\Phi_{j,b}$. This can be done by picking $\phi_1,\ldots,\phi_d \in \gf{p}$ randomly and uniformly. We then need to pick $\phi_0$ randomly and uniformly from the set
  $\{\phi_0 \cond \bool(\phi(j)) = b \}$. Given the result of inner product $\sum_{i\in [d]} z(j)_i \phi_i$ this can be done in $O(\log p)$ time as explained in the proof of Lemma~\ref{lem:rappor-communication-computation}. The computation of the inner product can be done in $d \cdot \tilde O(\log q) = \tilde O(\log k)$.
\end{proof}

The analysis of the running time of the aggregation algorithm Alg.~\ref{alg:agg-rappor-gen} follows from the discussion above.
\begin{lem}
\label{lem:rappor-decode-computation-gen}
The server-side of the histogram construction for generalized PI-RAPPOR (Alg.~\ref{alg:agg-rappor-gen}) can be done in time $\tilde O(n + k |\Phi| \log k)$.
\end{lem}

The running time depends on $|\Phi| \leq k q^2$ and thus also depends on the choice of $q$. The effect of the choice of $q$ on the utility of Algorithm~\ref{alg:agg-rappor-gen} is the same as the effect of the choice of $p$ on the utility of the simple PI-RAPPOR (given in Lemma~\ref{lem:pi-rappor-approx}). Thus we can assume that $q = O(\max\{e^\eps,1/\eps\})$. 

\end{document}